\documentclass{sig-alternate}
\makeatletter
\def\@copyrightspace{\relax}
\makeatother
\usepackage{mathtools,dsfont}
\usepackage{psfrag,bbm,graphicx,appendix}
\usepackage{algorithm,algorithmic}
\usepackage{comment,balance}
\usepackage{epstopdf}
\usepackage{epsfig}
\usepackage{multicol}
\usepackage{multirow, cite,url,enumitem }
\newtheorem{theorem}{Theorem}
\newtheorem{definition}{Definition}
\newtheorem{lemma}{Lemma}
\newtheorem{corollary}{Corollary}
\newtheorem{proposition}{Proposition}

\newcommand{\dd}{{\rm d}}
\newcommand{\oo}{{\rm o}}
\newcommand{\OO}{{\rm O}}
\newcommand{\be}{\begin{eqnarray}}
\newcommand{\ve}[1]{\mathbf{#1}}
\newcommand{\ee}{\end{eqnarray}}
\newcommand{\ben}{\begin{eqnarray*}}
\newcommand{\een}{\end{eqnarray*}}
\newcommand{\bfl}{\begin{flalign*}}
\newcommand{\efl}{\end{flalign*}}
\newcommand{\dref}[1]{(\ref{#1})}
\newcommand{\expect}[1]{{\mathbb E} \Bigl[ #1\Bigr]}
\newcommand{\expectp}[1]{{\mathbb E_{\pi}} \Bigl[ #1\Bigr]}
\newcommand{\expectpstar}[1]{{\mathbb E_{{\pi^\star}^{(\beta)}}} \Bigl[ #1\Bigr]}

\newcommand{\expectsk}[1]{{\mathbb E_{\ve{S(k)}}} \Bigl[ #1\Bigr]}
\newcommand{\expectsT}[1]{{\mathbb E_{\ve{S(kT)}}} \Bigl[ #1\Bigr]}

\newcommand{\expectqk}[1]{{\mathbb E_{\ve{Q(k)}}} \Bigl[ #1\Bigr]}

\newcommand{\calC}{{\mathcal C}}
\newcommand{\calA}{{\mathcal A}}
\newcommand{\calG}{{\mathcal G}}
\newcommand{\calN}{{\mathcal N}}
\newcommand{\calJ}{{\mathcal J}}
\newcommand{\calL}{{\mathcal L}}

\newcommand{\sumj}{\displaystyle\sum_{j\in \calJ}}
\newcommand{\sumA}{\displaystyle\sum_{A\in C^{(j)}(k)}}
\newcommand{\sumc}{\displaystyle\sum_{C \in \calC}}

\newcommand{\bias}{h}
\newcommand{\tw}{\tilde{w}^{(j)}_A}
\newcommand{\bjA}{b^{(j)}_A}

\begin{document}
\title{Scheduling Storms and Streams in the Cloud}
\numberofauthors{3}
\author{
\alignauthor{Javad Ghaderi\\}
\affaddr{Columbia University\\
New York, NY}\\
\affaddr{jghaderi@columbia.edu}
\and
\alignauthor{Sanjay Shakkottai\\}
\affaddr{University of Texas\\
Austin, TX}\\
\affaddr{shakkott@austin.utexas.edu}
\and
\alignauthor{R Srikant\\}
\affaddr{University of Illinois\\
Urbana, IL}
\affaddr{rsrikant@illinois.edu}
}

\maketitle
\begin{abstract}
Motivated by emerging big \textit{streaming} data processing paradigms (e.g., Twitter Storm, Streaming MapReduce), we investigate the problem of scheduling graphs over a large cluster of servers. Each graph is a job, where nodes represent compute tasks and edges indicate data-flows between these compute tasks. Jobs (graphs) arrive randomly over time, and upon completion, leave the system. When a job arrives, the scheduler needs to partition the graph and distribute it over the servers to satisfy load balancing and cost considerations. Specifically, neighboring compute tasks in the graph that are mapped to different servers incur load on the network; thus a mapping of the jobs among the servers incurs a cost that is proportional to the number of ``broken edges''. We propose a low complexity randomized scheduling algorithm that, without service preemptions, stabilizes the system with graph arrivals/departures; more importantly, it allows a smooth trade-off between minimizing average partitioning cost and average queue lengths. Interestingly, to avoid service preemptions, our approach does not rely on a Gibb's sampler; instead, we show that the corresponding limiting invariant measure has an interpretation stemming from a loss system.
\end{abstract}

%\category{G.3}{Probability and Statistics}{Queuing theory, } {Markov processes, }{Probabilistic algorithms}
%\category{C.2.1} {Network Architecture and Design}{Centralized networks, }{Distributed networks}
%\terms{Theory, Performance}

\begin{keywords}
Graph Partitioning, Dynamic Resource Allocation, Markov Chains, Probabilistic Algorithms
\end{keywords}
\section{Introduction}\label{sec:intro}

In recent years, a new computing model -- stream processing -- is
gaining traction for large-scale cloud computing systems. These
systems \cite{ZSS12,s4,storm,sphere} are driven by real time and
streaming data applications. For instance, consider the computation
needed to answer the question: {\em How may times does the hashtag
  ``\#sigmetrics2015'' appear in Twitter over the next two hours?} The
key feature here is that the data is not (yet) in a database; instead
it is appearing as and when people tweet this hashtag. Applications of
such stream computing are in many domains including social network
analytics and e-commerce.

To address such stream processing, the emerging computation model of
choice is that of graph processing. A computation is represented by a
graph, where nodes in the graph represent either data sources or data
processing (and operate sequentially on a stream of atomic data
units), and edges in the graph correspond to data flows between
nodes. To execute such computations, each node of a graph is mapped to
a machine (server/blade) in a cloud cluster (data center), and the
communication fabric of the cloud cluster supports the data flows
corresponding to the graph edges. A canonical example (and one of the
early leaders in this setting) is Twitter's Storm \cite{storm}, where
the (directed) graph is called a ``topology'', an atomic data unit is
a ``tuple'', nodes are called ``spouts'' or ``bolts'', and tuples flow along
the edges of the topology. We refer to \cite{LRL13} for additional
discussion.

From the cloud cluster side, there are a collection of machines
interconnected by a communication network. Each machine can
simultaneously support a finite number of graph nodes. This number is
limited by the amount of resources (memory/processing/bandwidth) that
is available at the machine; in Storm, these available
resources are called ``slots'' (typically order of ten to fifteen per
machine). Graphs (corresponding to new computations) arrive randomly over time to this cloud cluster, and
upon completion, leave the cluster.
At any time, the scheduling task at the cloud cluster is to
map the nodes of an incoming graph onto the free slots in machines to
have an efficient cluster operation. As an example, the default
scheduler for Storm is round-robin over the free slots; however, this
is shown to be inefficient, and heuristic alternatives have been been
proposed \cite{LRL13}.

In this paper we consider a queueing framework that models such
systems with graph arrivals and departures.  Jobs are graphs that are
{\em dynamically} submitted to the cluster and the scheduler needs to
to partition and distribute the jobs over the machines.
Once deployed in the cluster, the job (a computation graph) will
retain the resources for some time duration depending on the
computation needs, and will release the resources after the
computation is done (i.e., the job departs). The need for efficient
scheduling and dynamic graph partitioning algorithms naturally arises
in many parallel computing applications \cite{graphlab,giraph};
however, the theoretical studies in this area are very limited. To the
best of our knowledge, this is the first paper that develops models of
dynamic stochastic graph partitioning and packing, and the associated
low complexity algorithms with provable guarantees for graph-based
data processing applications.

From an algorithmic perspective, our low complexity algorithm has
connections to the Gibbs sampler and other MCMC (Monte Carlo Markov
Chain) methods for sampling probability distributions (see for example
\cite{bremaud}). In the setting of scheduling in wireless networks,
the Gibb's sampler has been used to design CSMA-like algorithms for
stabilizing the network
\cite{RSS09,shah2012randomized,liu2010towards,jiang2010distributed}.
However, unlike wireless networks where the solutions form independent
sets of a graph, there is no natural graph structure analog in the
graph partitioning. The Gibbs sampler can still be used in our setting
by sampling {\em partitions} of graphs, where, each site of the Gibbs
sampler is a unique way of partitioning and packing a graph among the
machines in the cloud cluster. The difficulty, however, is that there
are an exponentially large number of graph partitions, leading to a
correspondingly large number of queues.  The second issue is that a
Gibbs sampler potentially can interrupt ongoing service of jobs. The analog of a service interruption in
our setting is the migration of a job (graph) from one set of machines
to another in the cloud cluster. This is an expensive operation that
requires saving the state, moving and reloading on another set of
machines.

% One can still apply Gibbs sampler by considering a collection of
% sites, where each site corresponds to a particular way of partitioning
% and distributing a specific type of graph over the machines, and then,
% performing the site updates (i.e., whether an specific graph is
% partitioned according to that site or not) using the Gibbs sampler
% rule. The difficulty with this approach is that for any type of graph,
% there are an extremely large number of ways to do the
% partitioning. Considering various types of graphs, then one would need
% an exponentially large collection of sites, where each site needs to
% maintain an associated queue for the jobs.

A novelty of our algorithm is that \textit{we only need to maintain one
queue for each type of graph}. This substantial reduction is achieved
by developing an efficient method to explore the space of solutions in
the scheduling space. Further, our low complexity algorithm performs
updates at appropriate time instances without causing service
interruptions. In summary, our approach allows a smooth trade-off
between minimizing average partitioning cost and average queue sizes,
by using only a small number of queues, with low complexity, and
without service interruptions. As it will become clear later, the key
ingredient of our method is to minimize a modified energy function
instead of the Gibbs energy; specifically, the entropy term in the
Gibbs energy is replaced with the relative entropy with respect to a
probability distribution that arises in loss systems.

\subsection{Related Work}
\label{sec:related}

Dynamic graph scheduling occurs in many computing settings such as
Yahoo!'s S4 \cite{s4}, Twitter's Storm \cite{storm}, IBM's InfoSphere
Stream \cite{sphere}, TimeStream~\cite{QH13}, D-Stream~\cite{ZSS12},
and online MapReduce~\cite{onlinemapreduce10}. Current scheduling
solutions in this dynamic setting are primarily heuristic
\cite{LRL13,rychlyscheduling,chiang}.

The static version of this problem (packing a collection of graphs on
the machines on a one-time basis) is tightly related to \textit{the
  graph partitioning problem}~\cite{graphpartition13, wiki-graph},
which is known to be hard. There are several algorithms (either based
on heuristics or approximation bounds) available in the literature
\cite{feldmann2012balanced,andreev2006balanced,hendrickson1995multi,walshaw2000mesh,chiang}.

More broadly, dynamic bin packing (either scalar, or more recently
vector) has a rich history \cite{survey1,survey2}, with much
recent attention \cite{SZ13b,S13,GZS14}.  Unlike bin packing where
single items are placed into bins, our objective here is to pack
graphs in a dynamic manner.

% is bin packing that assigns items to bins and has been studied
% extensively for several decades, e.g., \cite{CoffmanGJ97,garey2,john,
%   survey1,survey2,online1,online2,online3,online4,online5,online6,best-lower,
%   best-upper,offline1,offline2,offline3,offline4,offline5}.

% Unlike bin packing that packs single items into bins, we aim to pack
% task graphs in a dynamic manner.

% ~
% \textbf{Javad: Please add
%   cites. Also some related work on vector bin packing, scaling dynamic
%   packing, optimization of inter-connected VMs (Chiang).}

%in contrast to batch processing systems (e.g., MapReduce),

% Most relevant to our paper is the reference \cite{LRL13} where the
% inefficiency of the default Storm scheduler is identified and a few
% heuristic solutions are proposed.

% However the key challenge
% here is that the site updates of the Gibb's sampler cannot be done
% arbitrarily at random instances; otherwise they would interrupt the
% service of ongoing jobs in the system.
% In the context of wireless networks, the Gibbs sampler in \cite{RSS09}
% could interrupt the ongoing transmissions of packets.

\subsection{Main Contributions}

We study the problem of partitioning and packing graphs over a cloud
cluster when graphs arrive and depart dynamically over time. The main
contributions of this work can be summarized as follows.

\begin{itemize}[leftmargin=*]
\item {\bf{ A Stochastic Model of Graph Partitioning.}} We develop a
  stochastic model of resource allocation for graph-based applications
  where either the computation is represented by a graph (Storm
  \cite{storm}, InfoSphere Stream \cite{sphere}) or the data itself
  has a graph structure (GraphLab \cite{graphlab}, Giraph
  \cite{giraph}). Most efforts have been on the systems aspects, while
  employing a heuristic scheduler for graph partitioning and
  packing. One of the contributions of this paper is the model itself
  which allows an analytical approach towards the design of efficient
  schedulers.

\item {\bf {Deficiencies of Max Weight-type Algorithms.}}  The dynamic
  graph partitioning problem can be cast as a network resource
  allocation problem; to illustrate we describe a frame-based Max Weight
  algorithm that can jointly stabilize the system and minimize packing
  costs. However, such Max Weight-type solutions have two
  deficiencies:

  (1) they involve periodically solving the static graph partitioning
  problem (NP-hard in general); thus there is little hope that
  this can be implemented in practice,

  (2) they require periodic reset of the system configuration to the
  Max Weight configuration; this interrupts a significant number of
  ongoing computations or services of the jobs in the system and
  require them to be migrated to new machines (which is expensive).

  % Such service interruptions are operationally very expensive as they
  % incur additional delay to the service or require the storage and
  % transfer of the state of interrupted jobs for future recovery.

\item {\bf {Low Complexity Algorithms without Service Interruptions.}}
  We develop a new class of low complexity algorithms, specifically
  targeted for the stochastic graph partitioning problems, and
  analytically characterize their delay and partitioning costs. In
  particular, the algorithms can converge to the optimal solution of
  the static graph partitioning problem, by trading-off delay and
  partitioning cost (a tunable parameter). Equally important, this
  class of algorithms do not interrupt the ongoing services in the
  system. The algorithms rely on creating and removing
  \textit{templates}, where each template represents a unique way of
  partitioning and distributing a graph over the machines. A key
  ingredient of the low complexity algorithms is that the decision to
  remove or add templates to the system is only made at the instances
  that a graph is submitted to the cluster or finishes its
  computation; thus preventing interruption of ongoing services.

  % We analyze and characterize the delay and partitioning cost
  % performance of an optimal frame-based algorithm. \item {\bf
  % {Technical.}}We propose a class of low complexity algorithms that
  % can provide accurate enough approximations to the frame-based
  % (optimal) algorithm while not interrupting the services of ongoing
  % jobs in the system. The key ingredient of the low complexity
  % algorithms is that the decision to remove or add templates to the
  % system are only made at the instances that a graph is submitted to
  % the cluster or finishes its computation.
\end{itemize}
\subsection{Notations}
%\textit{Notations:}
Some of the basic notations used in this
paper are the following. $|S|$ denotes the cardinality of a set $S$.
$A\backslash B$ is the set difference defined as $\{x \in A, x \notin
B\}$. $\mathds{1}\{x \in A\}$ is the indicator function which is $1$
if $x \in A$, and $0$ otherwise. $\mathds{1}_n$ is the $n$-dimensional vector of all ones. $\mathds{R}_+$ denotes the set of
real nonnegative numbers. For any two probability vectors $\pi,\nu \in
\mathds{R}^n$, the total variation distance between $\pi$ and $\nu$ is
defined as $\|\pi-\nu\|_{TV}=\frac{1}{2}\sum_{i=1}^n|\pi_i-\nu_i|$.
Further, the Kullback--Leibler (KL) divergence of $\pi$ from $\nu$ is defined as
$D_{\mathrm{KL}}(\pi\|\nu) = \sum_i \pi_i \, \log\frac{\pi_i}{\nu_i}$.
%Given two vectors $x,y \in \mathds{R}^n$, $x\leq y$ means $x_i \leq
%y_i$ componentwise; if in addition $x_i < y_i$ for at least one of the
%components, then $x <y$.
Given a stochastic process $z(t)$ which converges in distribution as
$t \to \infty$, we let $z(\infty)$ denote a random variable whose
distribution is the same as the limiting distribution. Given $x \in
\mathds{R}^n$, $x_{\min}=\min_i x_i$, $x_{\max}=\max_i x_i$.

% \textit{Organization:} We start with the description of system model
% and definitions in Section \ref{sec:system}. Section \ref{sec:prob} is
% devoted to formulation of static and dynamic graph partitioning
% problems. Section \ref{sec:high} discusses high complexity Max Weight
% solutions to the dynamic graph partitioning problem. In Section
% \ref{sec:low}, we propose two low complexity solutions and present the
% main results regarding their performance. The proof are presented in
% Section~\ref{sec:proofs}. Finally, Section \ref{sec:con} contains
% conclusions.

\section{System Model and Definitions}\label{sec:system}

\noindent{\em Cloud Cluster Model and Graph-structured Jobs:}
Consider a collection of machines $\calL$. Each machine $\ell \in \calL$ has a
set of slots $m_\ell$ which it can use to run at most $|m_\ell|$
processes in parallel (see Figure~\ref{fig:template}). These machines
are inter-connected by a communication network. Let $M=\sum_{\ell}|m_{\ell}|$ be the total number of slots in the cluster.

There is a collection of {\em jobs types} $\calJ$, where each job type
$j \in \calJ$ is described by a {\em graph} $\calG_j(V_j,E_j)$
consisting of a set of nodes $V_j$ and a set of edges $E_j$. Each
graph $\calG_j$ represents how the computation is split among the set
of nodes $V_j$. Nodes correspond to computation with each node
requiring a slot on some machine; edges represent data flows between
these computations (nodes).\\

\noindent{\em Job Arrivals and Departures:}
Henceforth, we use the word {\em job} and
{\em graph} interchangeably. We assume graphs of type $j$ arrive
according to a Poisson process with rate $\lambda_j$, and will remain
in the system for an exponentially distributed amount of time with
mean $1/\mu_j$. Node of the graph must be assigned to an empty slot on
one of the machines. Thus a graph of type $\calG_j$ requires a total
number of $|V_j|$ free slots ($|V_j|<M$). For each graph, data center needs to decide how to partition the
graph and distribute it over the machines.\\

\noindent{\em Queueing Dynamics:}
When jobs arrive, they can either
be immediately served, or queued and served at a later time. Thus, there
is a set of queues
$\ve{Q}(t)=(Q^{(j)}(t):\ j\in \calJ )$ representing
existing jobs in the system either waiting for service or receiving
service. Queues follow the usual dynamics:
\be
Q^{(j)}(t)=Q^{(j)}(0)+H^{(j)}(0,t)-D^{(j)}(0,t),
\ee
where $H^{(j)}(0,t)$ and $D^{(j)}(0,t)$ are respectively the number of
jobs of type $j$ arrived up to time $t$ and departed up to time
$t$.\\

\noindent{\em Job Partition Cost:}
For any job, we assume that the cost of data exchange between two
nodes that are inside the same machine is zero, and the cost of data
exchange between two nodes of a graph on different machines is
one. This models the cost incurred by the data center due to the total
traffic exchange among different machines. Note that this model is
only for keeping notation simple; in fact, if we make the cost of each
edge different (depending for instance on the pair of machines on
which the nodes are assigned, thus capturing communication network
topology constraints within the cloud cluster), there is minimal
change in our description below. Specifically, we only need to
redefine the appropriate cost in
(\ref{eqn:cost}), and the ensuing analysis will remain unchanged.\\

\noindent {\em Templates:}
%Analogous to job types
%and jobs, there is another important construct -- {\em template types}
%and {\em templates}. These keep track of the ways jobs have been
%partitioned among the machines; additionally they also keep track of
%the ways jobs can potentially be packed among the servers in the
%future.
% (See Figure~\ref{fig:template}).
%{\em Template types:}
An important construct in this paper is the concept of template.
Observe that for any graph $\calG_j$, there are
several ways (exponentially large number) in which it can be
\textit{partitioned} and \textit{distributed} over the machines (see
Figure~\ref{fig:template}). A {\em template} corresponds to one
possible way in which a graph $\calG_j$ can \textit{partitioned} and
\textit{distributed} over the machines (see
Figure~\ref{fig:template}). Rigorously, a template $A$ for graph
$\calG_j$ is an injective function $A: V_j \to \bigcup_{\ell \in
  \calL}m_\ell $ which maps each node of $\calG_j$ to a unique slot in
one of the machines. We use $\calA^{(j)}$ to denote the set of all
possible templates for graph $\calG_j$. Tying back to the cost
model, for $A \in \calA^{(j)}$, let
$b^{(j)}_A$ be the cost of partitioning $\calG_j$ according to
template $A$, then
\be
b^{(j)}_A=\sum_{(x,y)\in E_j}\mathds{1}\{A(x)\in m_\ell, A(y)\in
m_{\ell^\prime}, \ell \neq \ell^\prime\}. \label{eqn:cost}\\ \nonumber
\ee

%{\em Templates:}
\noindent {\em Configuration:}
While there are an extremely large number of templates possible for
each graph, only a limited number of templates can be present in the
system at any instant of time. This is because each slot can be used
by at most one template at any given time.

To track the collection of templates in the system, we let $C^{(j)}(t)
\subset \calA^{(j)}$ to be the set of existing templates of graphs
$\calG_j$ in the system at time $t$. The system configuration at each
time $t$ is then defined as \be \ve{C}(t)=\left( C^{(j)}(t);\ j\in
  \calJ \right).  \ee
%templates are actual instantiations that currently are in the
%system.
By definition, there is a template in the system corresponding to each job that is
being served on a set of machines. Further, when a new job arrives or
departs, the system can (potentially) create a new template that is a
pattern of {\em empty slots} across machines that can be ``filled''
with a specific job type (i.e., one particular graph topology). We
call the former as {\em actual templates}, and the latter as {\em
  virtual templates}. Further, when a job departs, the system can
potentially destroy the associated template.

\begin{figure}
  \centering
  \includegraphics[width=3in]{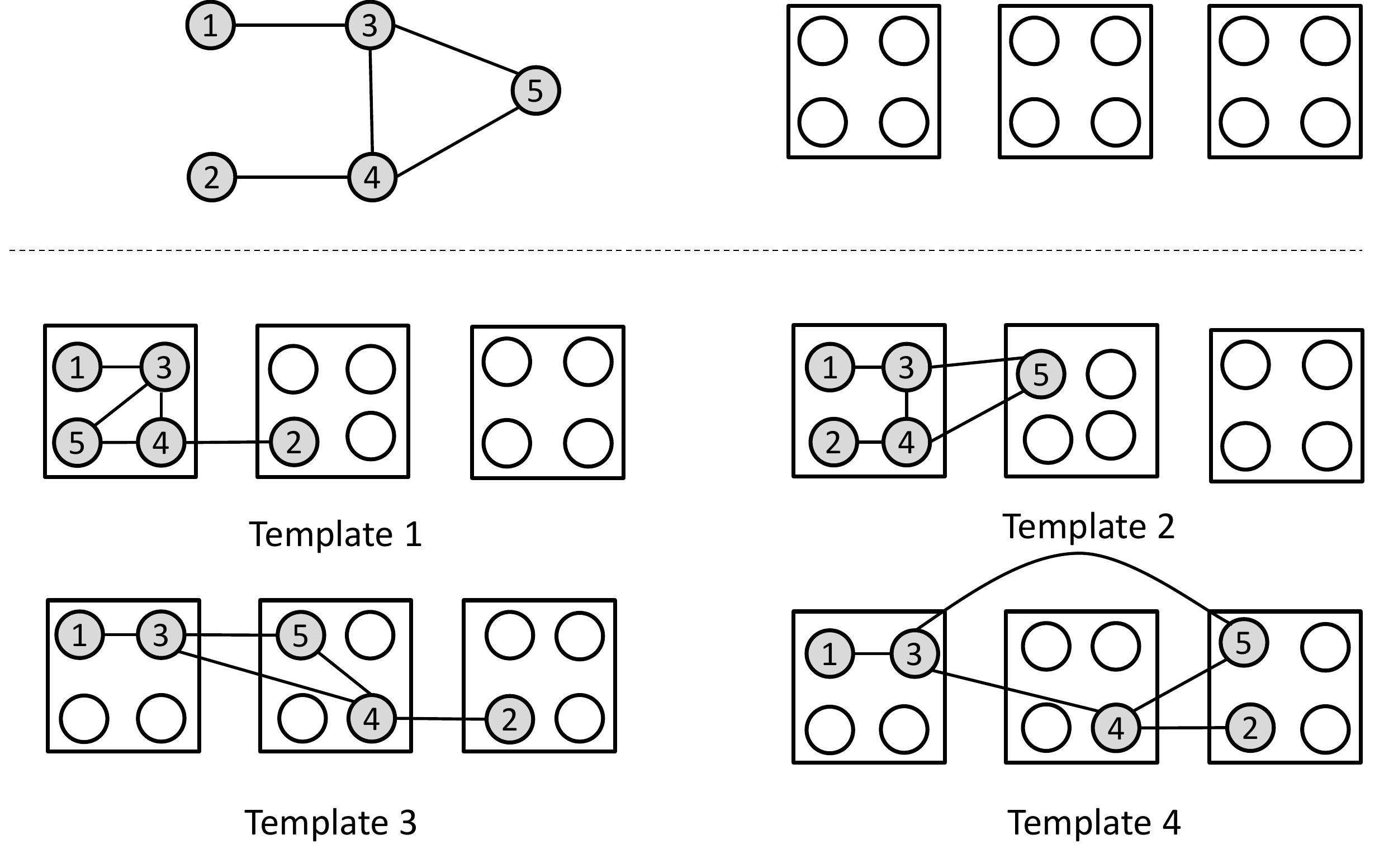}
  \caption{Illustrative templates for partitioning and distributing a
    five-node graph in a cluster of 3 servers, with each server having
    4 empty slots. In this stylized example, all the edges have unit
    ``breaking'' costs, i.e., two connected nodes being scheduled on
    different servers incurs a unit cost. The cost of partitioning the
    graph according to these templates is as follows:
    $b_{\text{Template 1}}= 1$, $b_{\text{Template 2}} = 2$,
    $b_{\text{Template 3}} = 3$, and $b_{\text{Template 4}} =
    4$.}\label{fig:template}
\end{figure}

% $N_j(t)$ denote the total number of templates in the system at time
% $t$ to accommodate graphs of type $\calG_j$, with $A^{(j)}_1,
% \cdots, A^{(j)}_{N_j}$ being the corresponding templates.  With a
% minor abuse of notation, sometimes we use $C^{(j)}(t)$ to be the
% corresponding incident vector in $\{0,1\}^{|\calA^{(j)}|}$, i.e., $
% C^{(j)}_A(t)=1$ if $A \in C^{(j)}(t)$.

%\noindent \underline{\em Configuration:}
The set of all possible configurations is denoted by $\calC$. Note that this collection is a
union of the actual and virtual templates.  Mathematically,
$C^{(j)}=C^{(j)}_a \cup C^{(j)}_v$ where $C^{(j)}_a$ is the set of
templates that contain \textit{actual} jobs of type $j$ and
$C^{(j)}_v$ is the set of \textit{virtual} templates, i.e., templates
that are reserved for jobs of type $j$ but currently do not contain
any such jobs.\\

\noindent {\em System State and Updates:}
Finally the system state at
each time is then given by:
\be
\ve{S}(t)=(\ve{Q}(t), \ve{C}(t)).
\ee
It is possible that $|C^{(j)}(t)|>Q^{(j)}(t)$ in which case not all the
templates in $C^{(j)}$ are being used for serving jobs, these unused templates are the virtual templates.

%Correspondingly, $\ve{C}_a(t)$ and $\ve{C}_v(t)$ are used to denote the incident vectors.

Define the operation $\ve{C} \oplus A^{(j)}$ as adding a feasible
template $A$ for graphs of type $\calG_j$ to the configuration
$\ve{C}$, thus $A$ will be added to $C^{(j)}$ while $C^{(j^\prime)}$
remains unchanged for $j^\prime \neq j$. Define $\calA^{(j)}(\ve{C})$
as the set of possible templates that can be used for adding a graph
$\calG_j$ when the configuration is $\ve{C}$. Clearly, $A \in
\calA^{(j)}(\ve{C})$ must be an injective function that maps graph
$G_j$ to the available slots that \textit{have not} been used by the
current templates in the system configuration, i.e.,
\[
A:V_j \to \left(\bigcup_{\ell\in\calL} m_\ell\right) \backslash
\left(\bigcup_{j^\prime\in \calJ} \bigcup_{A^\prime \in
    C^{(j^\prime)}} A^\prime(V_{j^\prime})\right).
\]

\section{Problem formulation}\label{sec:prob}
Given any stationary (and Markov) algorithm for scheduling arriving graphs, the system state evolves as an irreducible and aperiodic Markov chain. Our goal is to minimize the average partitioning cost, i.e.,
\be \label{problem}
\mbox{minimize } &&\expect{\sumj \sum_{A \in \calA^{(j)}}x_A(\infty)  b^{(j)}_{A}}\\
\mbox{subject to}&& \mbox{system stability} \nonumber
\ee
 where $x_A(\infty)$ is a random variable denoting the fraction of time that a template $A$ is used in steady state. The system stability in \dref{problem} means that the average delay (or average queue size) remains bounded. There is an inherent tradeoff between the average delay and the average partitioning cost. For more lenient delay constraints, the algorithm can defer the scheduling of jobs further until a feasible template with low partitioning cost becomes available.

Throughout the paper, let $\rho_j=\lambda_j/\mu_j$ be the load of graphs of type $\calG_j$.
% Let $\pi(C)$ denote the fraction of time that a configuration $C\in \calC$ is used in steady state.
%With minor abuse of notation, let $C^{(j)}(C)$ denote the number of templates of graphs $\calG_j$ in configuration $C$.
\begin{definition}[Capacity Region]
The capacity region of the system is defined as
\ben
\Lambda=\Big\{z \in \mathds{R}^{|\calJ|}_{+}: \exists \pi \mathrm{\ s.t.\ } z_j=\sum_{C \in \calC} \pi(C) |C^{(j)}|,\\
\sum_{C \in \calC} \pi(C)=1,\pi(C)\geq 0\Big\},
\een
where $|C^{(j)}|$ denotes the number of templates of graph $\calG_j$ in configuration $C$.
%Then the capacity region of the system is defined as
%\ben
%\Lambda_0 = \Big\{\rho \in \mathds{R}^{|\calJ|}_{+}: \exists z \in \lambda \mathrm{\ s.t.\ } \rho<z \Big \}.
%\een
%\be
%\Lambda=\left\{z \in \mathds{R}^{|\calJ|}_{+}: \sumj z_j |V_j| \leq \suml |m_\ell|
%\right\}.
%\ee
%and $|C^{(j)}|$ denotes the number of templates of graph $\calG_j$ in configuration $C$.
\end{definition}
By the definition, any load vector $z \in \Lambda$ can be supported by a proper time-sharing among the configurations, according to $\pi$.
%\ben
%\rho_j = \sum_{C\in \calC} \pi(C) |C^{(j)}|=\sum_{A \in \calA^{(j)}}x_A;\ \forall j \in \calJ.
%\een
Equivalently, for any $z \in \Lambda$, there exists an ${\ve{x}=[x_A: A \in \cup_{j} \calA^{(j)} ] }$ such that
\ben
z_j = \sum_{A \in \calA^{(j)}}x_A;\ j \in \calJ,
\een
where $x_A$ is the average fraction of time that template $A $ is used, given by
\ben
x_A=\sumc \pi(C)\mathds{1}(A \in C^{(j)});\ A \in \calA^{(j)},\ j \in \calJ.
\een
It follows from standard arguments that for loads outside $\Lambda$, there is no algorithm that can keep the queues stable.
%Any load vector $\rho \in \Lambda$ can be supported by a proper time-sharing among the configurations, i.e., for any $\rho \in \Lambda$, there exists some probability distribution $\pi$, $\sum_{C \in \calC} \pi(C)=1$, $\pi(C) \geq 0$, such that
%\ben
%\rho_j = \sum_{C\in \calC} \pi(C) |C^{(j)}|=\sum_{A \in \calA^{(j)}}x_A;\ \forall j \in \calJ.
%\een
%Here $|C^{(j)}|$ denotes the number of templates of graph $\calG_j$ in configuration $C$ and $$x_A=\sumc \pi(C)\mathds{1}(A \in C^{(j)})$$ is the average fraction of time that a template $A \in \calA^{(j)}$ is used. Let ${\ve{x}=(x_A: A \in \cup_{j} \calA^{(j)} ) }$.
Given the loads $\ve{\rho}=[\rho_j: j \in \calJ]$, we define an associated static problem.
\begin{definition}[Static Partitioning Problem]\label{def:stat}
\begin{align}
&\displaystyle \min_{\ve{x}} &G(\ve{x}):= \sumj \sum_{ A \in \calA^{(j)}}x_A  b^{(j)}_{A}\\
&\mathrm{subject\ to:} &  \sum_{ A \in \calA^{(j)}}x_A \geq \rho_j;\ j \in \calJ  \label{cons1}\\
&& x_A \geq 0;\  A \in \cup_{j} \calA^{(j)} \label{cons2}\\
&&\Big[\sum_{ A \in \calA^{(j)}}x_A;\ j \in \calJ\Big] \in \Lambda \label{cons3}
\end{align}
\end{definition}
The constraints \dref{cons1}-\dref{cons3} are the required stability conditions. In words, given $\rho_j$ graphs of type $\calG_j$, for $j \in \calJ$, the static partitioning problem is to determine how to partition and distribute the graphs over the servers so as to minimize the total partitioning cost.
For the set of supportable loads ($\rho \in\Lambda$), the static problem is feasible and has a finite optimal value.

If the loads $\rho_j$'s are known, one can solve the static partitioning problem and subsequently find the fraction of time $\pi(C)$ that each configuration $C$ is used. However, the static partitioning problem is a hard combinatorial problem to solve.

In the next sections, we will describe two approaches to solve the dynamic problem \dref{problem} that could converge to the optimal solution of the static partitioning problem, at the expense of growth in delay. The inherent tradeoff between delay and partitioning cost can be tuned in the algorithms. First, we describe a high complexity frame-based algorithm (based on traditional Max Weight resource allocation). Then, we proceed to propose our low complexity algorithm which is the main contribution of this paper.   
\section{High Complexity Frame-Based Algorithm}\label{sec:high}
The first candidate for solving the dynamic graph partitioning problem is to use a Max Weight-type algorithm, with a proper choice of weight for each configuration. However changing the configuration of the system can potentially interrupt a significant number of ongoing services of the jobs in the system. Such service interruptions are operationally very expensive as they incur additional delay to the service or require the storage and transfer of the state of interrupted jobs for future recovery. Hence, to reduce the cost of service interruptions, one can reduce the frequency of configuration updates. In particular, we describe a Frame-Based algorithm which updates the configuration once every $T$ time units. As expected, a smaller value of $T$ could improve the delay and the partitioning cost of the algorithm at the expense of more service interruptions. The description of the algorithm is as follows.

%\vspace{0.15in}
%\hrule
%\vspace{0.03in}
%Basic Low Complexity Algorithm
%\captionof{algorithm}{Frame-Based Algorithm}
%\subsubsection*{Frame-Based Algorithm:}
\begin{algorithm}[H]
\caption{Frame-Based Algorithm}
%\vspace{0.03in}
%\hrule
%\vspace{0.03in}
\begin{algorithmic}[1]
\STATE The configuration is changed at the epochs of cycles of length $T$. At the epoch of the $k$-th cycle, $k=0,1,\cdots$, choose a configuration $ C^\star(kT)$ that solves
\be \label{eq:maxweightframe}
 \max_{C\in \calC} \sumj \sum_{ A \in C^{(j)}}\left( \alpha f\left(Q^{(j)}(kT)\right)-b^{(j)}_{A}\right).
\ee
If there are more than one optimal configuration, one of them is chosen arbitrarily at random. $\alpha >0$ is a fixed parameter and $f$ is a concave increasing function.

\STATE The configuration $C^\star(kT)$ is kept fixed over the interval $[kT,(k+1)T)$ during which jobs are fetched from the queues and are placed in the available templates. It is possible that at some time, no jobs of type $j$ are waiting to get service, in which case some of the templates in ${C^\star}^{(j)}(kT)$ might not be filled with the actual jobs. These are virtual templates which act as place holders (tokens) for future arrivals.
\end{algorithmic}
%\vspace{0.03in}
%\hrule
%\vspace{0.15in}
\end{algorithm}
The algorithm essentially needs to find a maximum weight configuration at the epochs of cycles, where the weight of template $A$ for partitioning graph $\calG_j$ is
\ben
 w^{(j)}_A(t)= \alpha f(Q^{(j)}(t))-b^{(j)}_A.
\een
The parameter $\alpha $ controls the tradeoff between the queue size and the partitioning cost of the algorithm. For small values of $\alpha$, the algorithm defers deploying the graphs in favor of finding templates with smaller partitioning cost. For larger values of $\alpha$, the algorithm gives a higher priority to deployment of job types with large queue sizes.
%$f$ is some concave (need not be strictly concave) increasing function including a linear function.

The optimization \dref{eq:maxweightframe} is a hard combinatorial problem, as the size of the configuration space $\calC$ might be exponentially large, thus hindering efficient computation of the max weight configuration. Theorem~\ref{th1} below characterizes the the inherent tradeoff between the average queue size and the average partitioning cost.
 %The templates that are filled with jobs are called \textit{actual templates}.
\begin{theorem}\label{th1}
Suppose $\rho(1+\delta^\star) \in \Lambda$ for some $\delta^\star>0$. The average queue size and the average partitioning cost under the Frame-Based algorithm are
\begin{align*}
&\expect{\sumj f(Q^{(j)}(\infty))} \leq \frac{(B_1+B_2T)+(1+\delta^\star)G(x^\star)/\alpha}{ \delta^\star \rho_{\min}}+B_1T\\
&\expect{\sumj \sum_{A \in \calA^{(j)}}x_A(\infty)  b^{(j)}_{A}} \leq G(x^\star)+\alpha(B_1+B_2T)
\end{align*}
where $x^\star$ is the optimal solution to the static partitioning problem, $\rho_{\min}=\min_j \rho_j$, and $B_1, B_2$ are constants.
\end{theorem}
Hence, as $\alpha \to 0$, the algorithm yields an $\alpha$-optimal partitioning cost, and an $\OO\Big(f^{-1}(\frac{1}{\alpha})\Big)$ queue size. Also as expected, infrequent configuration updates could increase the delay and partitioning cost by multiples of $T$. The proof of Theorem~\ref{th1} follows from standard Lyapunov arguments and can be found in the appendix.

\section{Low Complexity Algorithms without Service interruptions}\label{sec:low}
 In this section, we develop a low complexity algorithm that can be used to solve \dref{problem} without interrupting/migrating the ongoing services. Before describing the algorithm, we first introduce a (modified) \textit{weight} for each template. Given the vector of queue sizes $\ve{Q}(t)$, and a concave increasing function $f:\mathds{R}_+\to \mathds{R}_+$, the weight of template $A \in \calA^{(j)}$, $j\in \calJ$, is defined as
  \be \label{eq:tildeweight}
  \tilde{w}^{(j)}_A(t)=\alpha f^{(j)}\Big(\ve{h}+\ve{Q}(t)\Big)-\bjA,
  \ee
  where $ f^{(j)}:\mathds{R_+}^{|\calJ|}\to \mathds{R}_+$ is
  \be \label{eq:tildeweight2}
  f^{(j)}(x)=\max\Big\{f(x_j), \frac{\epsilon}{8M} f(x_{\max})\Big\};\ x_{\max}=\max_{j \in \calJ} x_j,
  \ee
%  and $f:\mathds{R}_{++}\to \mathds{R}_{+}$ is the \textit{log}-type function
%  \ben
%  f(y)=\log^{1-b}(y);\ 0<b<1.
%  \een
where $\ve{h}=h \mathds{1}_{|\calJ|}$, and $\alpha, h \in \mathds{R}_+$, and $\epsilon \in (0,1)$ are the parameters of the algorithm.

At the instances of job arrivals and departures, the algorithm makes decisions on the templates that are added to/removed from the system configuration. It is important that the addition/removal of templates by the algorithm \textit{does not} disrupt the ongoing service of existing jobs in the configuration.

The low complexity algorithm is a randomized algorithm in which the candidate template to be added to the configuration is chosen randomly among the set of feasible templates. In particular, the following \textit{Random Partition Procedure} is used as a subroutine in our low complexity algorithm.
%\vspace{0.15in}
%\hrule
%\vspace{0.03in}
%\subsubsection*{Random Partition Procedure:}
\begin{algorithm}[H]
\caption{Random Partition Procedure}
%\vspace{0.03in}
%\hrule
%\vspace{0.03in}
%\begin{algorithm}[h]
%\caption{Random Partition Procedure}
\textbf{Input:} current configuration $C$, and a graph $\calG(V,E)$, $V=\{v_1,\cdots,v_{|V|}\}$\\
\textbf{Output:} a virtual template $A$ for distributing $\calG$ over the machines.
\begin{algorithmic}[1]
\STATE $k \gets 1$
\STATE slot-available $\gets 1$
\WHILE{$k \leq |V|$ and slot-available}
\IF{there are no free slots available on any of the machines}
\STATE slot-available $\gets 0$; $A\gets \varnothing$
\ELSE
\STATE place $v_k$ uniformly at random in one of the free slots
\STATE $A(v_k)=$ index of the slot containing $v_k$
\STATE $k \gets k+1$
\ENDIF
\ENDWHILE
\end{algorithmic}
%\vspace{0.03in}
%\hrule
%\vspace{0.15in}
\end{algorithm}
%\vspace{0.15in}
%\hrule
%\vspace{0.03in}
%Basic Low Complexity Algorithm
%\subsubsection*{Dynamic Graph Partitioning $\mathrm{(DGP)}$ Algorithm:}
When a random template is generated according to Random Partition Procedure, the decision to keep or remove the template is made probabilistically based on the weight of the template. The description of the low complexity algorithm (called \textit{Dynamic Graph Partitioning (DGP) algorithm}) is as follows.
In the the description, $\beta >0$ is a fixed parameter.
\begin{algorithm}[H]
\caption{Dynamic Graph Partitioning $\mathrm{(DGP)}$}
%\vspace{0.03in}
%\hrule
%\vspace{0.03in}
\noindent\textbf{Arrival instances.}
\noindent Suppose a graph (job) $\calG_j$ arrives at time $t$, then:
\begin{algorithmic}[1]
\STATE This job is added to queue $Q^{(j)}$.
\STATE A virtual graph $\calG_j$ is randomly distributed over the machines, if possible, using \textit{Random Partition Procedure}, which creates a virtual template $A^{(j)}$ for distributing a graph $\calG_j$ over the machines with some partitioning cost $b^{(j)}_A$. Then, this virtual template is added to the current configuration with probability
$\frac{\exp(\frac{1}{\beta}\tw(t^+))}{1+\exp(\frac{1}{\beta}\tw(t^+))},$ otherwise, it is discarded and the configuration does not change. The virtual templates of type $j$ leave the system after an exponentially distributed time duration with mean $1/\mu_j$.
\STATE If there is one or more virtual templates available for accommodating graphs of type $\calG_j$, %and there are jobs in $Q^{(j)}$ waiting for service,
a job from $Q^{(j)}$ (e.g., the head-of-the-line job) is placed in one of the virtual templates chosen arbitrarily at random. This converts the virtual template to an actual template.
\end{algorithmic}

\noindent\textbf{Departure instances.} Suppose a departure of a (virtual or actual) template $A^{(j)}$ occurs at time $t$, then:
\begin{algorithmic}[1]
\STATE If this an actual template, the job departs and queue $Q^{(j)}(t^+)$ is updated.
\STATE A virtual template of the same type $A^{(j)}$ is added back to the configuration with probability $\frac{\exp(\frac{1}{\beta}\tw(t^+))}{1+\exp(\frac{1}{\beta}\tw(t^+))}.$
\STATE If a virtual template for accommodating a graph $\calG_j$ is available in the system, and there are jobs in $Q^{(j)}(t^+)$ waiting to get service, a job from $Q^{(j)}$ (e.g. the head-of-the line job) is placed in one of the virtual templates chosen arbitrarily at random. This converts the virtual template to an actual template.
\end{algorithmic}
%\vspace{0.03in}
%\hrule
%\vspace{0.15in}
\end{algorithm}
To simplify the description, we have assumed that the system starts from empty initial configuration and empty queues but this is not necessary for the results to hold.
We emphasize that the $\mathrm{DGP}$ algorithm does not interrupt the ongoing services of existing jobs in the system. The following theorem states our main result regarding the performance of the algorithm.
\begin{theorem}\label{th3}
Suppose $\rho(1+\delta^\star)\in \Lambda$ for some $0<\delta^\star<1$. Consider the Dynamic Graph Partitioning (DGP) algorithm with function
$$f(x)=\log^{1-b}(x);\ b \in (0,1),$$ and parameters
$$\alpha \leq \beta<1;\ \epsilon \leq \delta^\star;\ h\geq \exp\Big(C_0\frac{1}{\beta}(\frac{1}{\epsilon})^{\frac{2-b+1/b}{1-b}}\Big),$$ where $C_0$ is a large constant independent of all these parameters. Then the average queue size and the average partitioning cost under the DGP algorithm are
\begin{flalign*}
&\sumj\expect{f(Q_{j}(\infty))} \leq \frac{2}{\rho_{min}\delta^\star } \Big(\hat{K}_2+\hat{K}_3-\frac{\beta}{\alpha} \log \gamma_{min}+\\
&\frac{1}{\alpha}(1+\delta^\star/2)G(x^\star)+\frac{\epsilon}{\alpha} b_{max}\Big),\\
&\expect{\sumj \sum_{A \in \calA^{(j)}}x_A(\infty)  b^{(j)}_{A}}\leq G(x^\star)+\alpha(\hat{K}_2+\hat{K}_3)\\
&- \beta \log\gamma_{min}+\epsilon b_{max},
\end{flalign*}
where $x^\star$ is the optimal solution to the static partitioning problem, $\rho_{min}=\min_j \rho_j$, $\hat{K}_2\leq f^\prime(\bias)(M+\sum_j \rho_j)$ and $\hat{K}_3 \leq f(M + \bias )M $, and $\gamma_{\min}$, and $b_{\max}$ are constants.
\end{theorem}

We would like to point out that in the above theorem the bounds are explicit for any choices of $\alpha, \beta, \epsilon, h$. The constant $\gamma_{\min}$ is $\min_C \gamma_C $ for a distribution $\gamma$ to be defined in \dref{gammaform} and has a loss-system interpretation (see Step 1 in the Proof of Theorem~\ref{th3}), and $b_{\max}$ is the maximum partitioning cost of any job type (which is obviously less than $M^2$).

The parameter $h$ is called the \textit{bias} and adds an offset to the queues to ensure the algorithm operates near the optimal point at (effectively) all times. The parameter $\beta$ has the similar role as the temperature in Gibbs sampler. As $\beta \to 0$, in steady state, the algorithm generates configurations that are closer to the optimal configuration, however at the expense of growth in queue sizes. We refer to Section \ref{sec:proofs} for the proof and also more insight into the operation of the algorithm.

The following corollary gives an interpretation of the result for a particular choice of the parameters.
\begin{corollary}\label{cor-th}
Choose $\alpha=\beta^{2}$, $h=\exp\Big((\frac{1}{\beta})^{1/(1-b)}\Big)$, $\epsilon =\beta^{b^2/4}$, then as $\beta \to 0$,
\begin{flalign*}
&\sumj\expect{f(Q_{j}(\infty))} \leq \Theta((\frac{1}{\beta})^{2}),\\
&\expect{\sumj \sum_{A \in \calA^{(j)}}x_A(\infty)  b^{(j)}_{A}}\leq G(x^\star)+\Theta(\beta^{b^2/4}).
\end{flalign*}
\end{corollary}
The corollary above demonstrates how the choice of $\beta$ controls the tradeoff between approaching the optimal partitioning cost and the queueing performance.

\textit{Remark 1. Comparison with CSMA:}
In the setting of scheduling in wireless networks, the Gibbs sampler has been used to design CSMA-like algorithms for
stabilizing the network \cite{RSS09,shah2012randomized,liu2010towards,jiang2010distributed}. Our algorithm is different from this line of work in three fundamental aspects:

(i) Not relying on Gibbs sampler: Unlike wireless networks where the solutions form independent
sets of a graph, there is no natural graph structure analog in the graph partitioning. The Gibbs sampler (and CSMA) can still be used in our setting by sampling {\em partitions} of graphs, where, each site of the Gibbs sampler is a unique way of partitioning and packing a graph among the
machines. The difficulty, however, is that there are an exponentially large number of graph partitions, leading to a
correspondingly large number of queues for each type of graph. A novelty of our algorithm is that \textit{we only need to maintain one
queue for each type of graph}. This substantial reduction is achieved
by using \textit{Random Partition Procedure} for exploring the space of solutions. This leads to minimizing a modified energy function instead of the Gibbs energy; specifically, the entropy term in the
Gibbs energy is replaced with the relative entropy with respect to a probability distribution that arises in an associated loss system (see Step 1 in Section~\ref{sec:proofs}).

(iii) No service interruptions: Our low complexity algorithm performs updates at appropriate time instances without causing service interruptions.

(iii) Adding bias to the queues: The queue-based CSMA algorithms are concerned with stability which pertains to the behavior of the algorithm for large queue sizes. This is not sufficient in our setting because we are not only concerned with stability, but more importantly with the optimal (graph partitioning) cost of the system. The bias $h$ boosts the queue sizes artificially to ensure that the system operates effectively near the optimal point at \textit{all queue sizes}. Without the bias, when the queue sizes are small, the optimal cost of the algorithm could be far from optimal.

\textit{Remark 2. An Alternative Algorithm:}
An alternative description of the algorithm is possible using a dedicated Poisson clock for each queue (independent of arrivals) where the template decisions are made at the ticks of the dedicated clocks. We have presented this alternative algorithm in the appendix.
%In particular, by choosing $b$ close to $1$, we can  achieve an average partitioning cost which is roughly $\beta$-optimal, with the average queue size growing roughly as $\OO(f^{-1}(\frac{1}{\beta}))$. 
\section{Proofs}\label{sec:proofs}
In this section, we present the proof of of Theorem~\ref{th3}. Before describing the proof outline, we make the following definition.

%\subsection{Proof of Theorem~\ref{th3}}
\textit{Definition: $\overline{\mathrm{DGP}}(\tilde{W})$.}
Consider the dynamic graph partitioning algorithm with fixed weights $\tilde{W}=[\tw; A\in \calA^{(j)}, j \in \calJ]$, namely, when weights are not chosen according to \dref{eq:tildeweight} but they are simply some fixed numbers all the time. With minor abuse of notations, we use $\overline{\mathrm{DGP}}(\tilde{W})$ to denote this algorithm that uses weights $\tilde{W}$ all the time. Description of $\mathrm{DGP(}\tilde{W})$ is exactly the same as the dynamic partitioning algorithm, except that at arrival/departure instance at time $t$, the decision to add/keep a virtual template $A^{(j)}$ is made according to probability $\frac{\exp(\frac{1}{\beta}\tw)}{1+\exp(\frac{1}{\beta}\tw)},$ independently of ${Q}(t)$.

\textit{Proof Outline.} The proof of Theorem ~\ref{th3} has three steps:
\begin{itemize}
\item[Step 1:] We analyze the steady-state distribution of configurations under $\overline{\mathrm{DGP}}(\tilde{W})$ with fixed weights $\tilde{W}$, and show that for small values of $\beta$, $\overline{\mathrm{DGP}}(\tilde{W})$ will generate configurations which are ``close'' to the max weight configuration, when the template weights are per $\tilde{W}$.
\item[Step 2:] We show that when weights are chosen according to \dref{eq:tildeweight}, although the weights $\tilde{W}(t)$ are time-varying, , the distribution of configurations in the system will be ``close'' to the corresponding steady-state distribution of $\overline{\mathrm{DGP}}(\tilde{W}(t))$, for \textit{all times} $t$ long enough. We show that such ``time-scale decomposition'' holds under the suitable choice of the bias $h$ and the function $f$.
\item[Step 3:] Finally, we stitch the dynamics of queues and configurations together through Lyapunov optimization method to compute the queueing and partitioning cost of our algorithm.
\end{itemize}
\subsubsection*{Step 1: Steady-State Analysis of $\overline{\mathrm{DGP}}(\tilde{W})$}
Under $\overline{\mathrm{DGP}}(\tilde{W})$, the configuration of the system evolves as a ``time-homogeneous'' Markov chain over the state space $\calC$. Note that from the perspective of evolution of configuration in the system, we do not need to distinguish between virtual and actual templates, since transition rates from any configuration $C$ do not depend on whether the templates in $C$ are actual or virtual. To see this, consider any virtual template of graphs $\calG_j$ in $C(t)$. No matter if the virtual template is filled with an actual job or not, the residual time until the departure of this template is still exponential with rate $\mu_j$, due to the memoryless property of exponential distribution and because both virtual templates and jobs have exponential service times with the same mean $1/\mu_j$. The following proposition states the main property of $\overline{\mathrm{DGP}}(\tilde{W})$.
\begin{proposition}\label{th2}
Consider the $\overline{\mathrm{DGP}}(\tilde{W})$ with fixed weights $\tilde{W}=[\tw; A\in \calA^{(j)}, j \in \calJ]$. Then in steady state, the distribution of configurations $\pi$ will solve the following optimization problem
\be \label{eq:mainprop1}
\max_{\pi\in \mathds{R}_+^{|\calC|},\sum_{C \in \calC}\pi_C=1} \expectp{\sum_{j\in \calJ} \sum_{A \in C^{(j)}} \tw}-\beta D_{KL}(\pi \parallel \gamma),
%&\mbox{}& \sum_{C \in \calC}\pi_C=1,
\ee
where $D_{KL}(\cdot \parallel \cdot)$ is the KL divergence of $\pi$ from the probability distribution $\gamma$, where
\be \label{gammaform}
\gamma_C=\frac{1}{Z_\gamma}\left(\sum_\ell |m_l|-\sum_j|C^{(j)}||V_j|\right) !\prod_{j}{\rho_j}^{|C^{(j)|}}, C\in\calC
\ee
and ${Z_\gamma}$ is the normalizing constant.
\end{proposition}
Before describing the proof of Proposition~\ref{th2}, we briefly highlight the main features of $\overline{\mathrm{DGP}}(\tilde{W})$  algorithm:
\begin{itemize}
\item[(i)] The algorithm does not interrupt the ongoing services of existing jobs in the system and does not require dedicated computing resources.

\item[(ii)] The algorithm is different from Gibbs sampler as it does not maximize the Gibbs energy. The entropy term $H(\pi)$ in the Gibbs energy has been replaced by the relative entropy $D_{KL}(\pi,\gamma)$.

\item[(iii)] The distribution $\gamma$ has the interpretation of the steady-state distribution of configurations in an \textit{associated loss system} defined as follows: at arrival instances, the arriving graph is randomly distributed over the machines if possible (according to Random Partition Procedure), otherwise it is dropped; at the departure instances, the job (and hence its template) leaves the system.
\end{itemize}
\begin{proof}[of Proposition \ref{th2}]
Consider the maximization problem
\ben
\max_{\{\pi(C)\}}  &F^{(\beta)}(\ve{\pi}) \label{opt2}\\
\mbox{subject to }&\sum_{C\in \calC} \pi(C)=1 \nonumber \\
&\pi(C) \geq 0,\ \forall C \in \calC.
\een
with function $F^{(\beta)}(\pi)$ as in \dref{eq:mainprop1}, which is
\ben
F^{(\beta)}(\ve{\pi})&=&\sum_{\calC}\pi(C)\sum_{j\in \calJ} \sum_{A \in C^{(j)}} \tw-\beta \sum_\calC \pi(C)\log \pi(C)\\
&&+\beta \sum_\calC \pi(C)\log \gamma(C).
\een
Notice that $F^{(\beta)}(\ve{\pi})$ is strictly concave in $\ve{\pi}$. The lagrangian is given by $L(\ve{\pi},\eta)=F^{(\beta)}(\ve{\pi})+\eta (\sum_\calC \pi(C)-1)$ where $\eta \in \mathds{R}$ is the lagrange multiplier. Taking $\partial L/\partial \pi(C)=0$ yields
\ben \label{eq:pinonneg}
\pi(C)=\exp(-1+\frac{\eta}{\beta})\gamma(C)\exp(\frac{1}{\beta}\sum_{j \in \calJ} \sum_{A \in C^{(j)}} \tw);\ C \in \calC,
\een
which is automatically nonnegative for any $\eta$. Hence, by KKT conditions $(\ve{\pi}^\star,\eta^\star)$ is the optimal primal-dual pair if it satisfies $\sum_\calC \pi^\star(C)=1$. Thus the optimal distribution $\pi^\star$ is
\be \label{expform}
\pi^\star (C)=\frac{1}{Z_\beta}\gamma(C)\exp(\frac{1}{\beta}\sum_{j \in \calJ} \sum_{A \in C^{(j)}} \tw).
\ee
where $Z_\beta$ is the normalizing constant.

Next we show that the $\overline{\mathrm{DGP}}(\tilde{W})$ algorithm indeed produces the steady-state distribution \dref{expform} with the choice of $\gamma$ in \dref{gammaform}, by checking the detailed balance equations. Consider a template $A^{(j)}$ for graphs of type $\calG_j$. The detail balanced equation for the pair $C$ and $C \oplus A^{(j)}$, such that $C \oplus A^{(j)} \in \calC$, is given by
\ben
\pi(C \oplus A^{(j)})\mu_j \frac{1}{1+e^{\frac{1}{\beta}w_A^{(j)}}}=\pi(C) \frac{\lambda_j}{|\calA^{(j)}(C)|}\frac{e^{\frac{1}{\beta}\tw}}{1+e^{\frac{1}{\beta}\tw}}.
\een
The left-hand-side is the departure rate of (virtual or actual) template $A^{(j)}$ from the configuration $C \oplus A^{(j)}$. The right-hand-side is the arrival rate of (actual or virtual) graphs $\calG_j$ to the configuration $C$ that are deployed according to template $A^{(j)}$ chosen uniformly at random from $\calA^{(j)}(C)$ (Recall that Random Partition Procedure used in the algorithm selects a template $A \in \calA^{(j)}(C)$ uniformly at random).
Thus the detailed balanced equation is simply
\be
\pi(C \oplus A^{(j)})&=&\pi(C)\frac{\rho_j}{{|\calA^{(j)}(C)|}}e^{\frac{1}{\beta}\tw}.
\ee
Noting that
$$|\calA^{(j)}(C)|=\binom{\sum_\ell |m_\ell|-\sum_j |C^{(j)}||V_j|}{|V_j|}|V_j|\, !,$$
it is then easy to see that $\dref{expform}$ with $\gamma$ as in \dref{gammaform}, indeed satisfies the detailed balance equations, and the normalizing condition that $\sum_C \pi(C)=1$.
%show by induction that
%\be \label{eq:pi0}
%\pi(C)&=&\pi(0)\frac{(\sum_\ell |m_\ell|-\sum_j|C^{(j)}||V_j|)\, !}{(\sum_\ell |m_\ell|)\, !}\prod_{j}{\rho_j}^{|C^{(j)}|} \nonumber \\
%&& \times \exp \Big(\frac{1}{\beta}\sumj \sum_{ A \in C^{(j)}} \tw\Big)
%\ee
%for $\pi(0)\neq 0$ is a valid solution where $\pi(0)$ is chosen to ensure $\sum_{C \in \calC} \pi(C)=1$.
%This shows that the steady-state distribution generated by the algorithm is indeed given by (\ref{expform}), with the choice of $\gamma(C)$ and $Z_{\beta}=Z_{\gamma}\pi(0)$ as defined in (\ref{gammaform}).
This concludes the proof.
\end{proof}
The parameter $\beta$ has the similar role as the temperature in Gibbs sampler. As $\beta \to 0$, in steady state, the $\overline{\mathrm{DGP}}(\tilde{W})$ algorithm generates configurations that are closer to the optimal configuration with maximum weight
\be
\tilde{W}^\star=\max_{C \in \calC} \sum_{j \in \calJ} \sum_{A \in C^{(j)}} \tw.
\label{eq:maxweight}
\ee
%however, the time required to reach close to steady state grows as a multiple of $\exp(\frac{1}{\beta})$.
The following corollary contains this result.
\begin{corollary}\label{cor1}
Let $F^{(\beta)}({\pi^\star}^{(\beta)})$ be the optimal objective function in \dref{eq:mainprop1}. The algorithm $\overline{\mathrm{DGP}}(\tilde{W})$ is asymptotically optimal in the sense as $\beta \to 0$, $F^{(\beta)}({\pi^\star}^{(\beta)}) \to \tilde{W}^\star$. Moreover, for any $\beta>0$,
\ben
\expectpstar{\sum_{j \in \calJ} \sum_{A \in C^{(j)}} \tw}\geq \max_{C \in \calC} \sum_{j \in \calJ} \sum_{A \in C^{(j)}} \tw+\beta \min_{C\in \calC} \log{\gamma_C}
\een
\end{corollary}
\begin{proof}[of Corollary \ref{cor1}]
Let $\tilde{C}^\star$ be the maximizer in \dref{eq:maxweight}. As a direct consequence of Proposition \ref{th2},
\ben
\expectpstar{\sum_{j \in \calJ} \sum_{A \in C^{(j)}} \tw}-\beta D({\pi^\star}^{(\beta)} \parallel \gamma) \geq \tilde{W}^\star -\beta D(\delta_{\tilde{C}^\star} \parallel \gamma).
\een
Since $D(\upsilon \parallel \gamma)\geq 0$, for any distribution $\upsilon$,
\ben
\expectpstar{\sum_{j \in \calJ} \sum_{A \in C^{(j)}} \tw} &\geq& \tilde{W}^\star -\beta D(\delta_{C^\star} \parallel \gamma)\\
& = & \tilde{W}^\star +\beta \log{\gamma_{\tilde{C}^\star}}\\
& \geq & \tilde{W}^\star +\beta \min_{C \in \calC} \log \gamma_C.
\een
\end{proof}
\subsubsection*{Step 2: Time-Scale Decomposition for $\mathrm{DGP}(\tilde{W}(t))$.}
Recall that $\overline{\mathrm{DGP}}(\tilde{W}(t))$ denotes the algorithm that uses the weight $\tilde{W}(t)$ at all times $s\geq 0$.
With minor abuse of notation, we use $\mathrm{DGP}(\tilde{W}(t))$ to denote the Dynamic Graph Partitioning algorithm (Section~\ref{sec:low}) and its associated time-inhomogeneous Markov chain over the space of configurations $\calC$. The weights $\tilde{W}(t)$ are time-varying (because of the queue dynamics), however the $\mathrm{DGP}(\tilde{W}(t))$ algorithm can still provide an adequately accurate approximation to the optimization \dref{eq:maxweight} at each time, for proper choices of function $f$ and the bias $h$.

Roughly speaking, for the proper choices of $f$ and $h$, $f(h+Q^{(j)}(t))$ will change adequately slowly with time such that a time-scale separation occurs, i.e., convergence of Markov chain $\overline{\mathrm{DGP}}(\tilde{W}(t))$ to its steady state distribution ${\pi_t} ^{(\beta)}$ will occur at a much faster time-scale than the time-scale of changes in $f(h+Q^{(j)}(t))$ (and thus in the weights). Hence, the probability distribution of configurations under $\mathrm{DGP}(\tilde{W}(t))$ will remain ``close'' to ${\pi_t} ^{(\beta)}$ (the steady state distribution of configurations under $\overline{\mathrm{DGP}}(\tilde{W}(t))$). The proof of such a time-scale separation follows from standard arguments in e.g., \cite{RSS09,GS10,ghaderiflow}.

We first uniformize (e.g.~\cite{lippman75,puterman09}) the continuous Markov chain $\ve{S}(t)=(\ve{C}(t),\ve{Q}(t))$ by using a Poisson clock $N_\xi(t)$ of rate
\be \label{eq:xi}
\xi=2\left(\sum_j\lambda_j+M \sum_j\mu_j\right).
\ee
Let $\ve{S}[k]=(\ve{C}[k],\ve{Q}[k])$ be the corresponding jump chain of the uniformized chain. Note that $\ve{S}[k]$ is discrete time and at each index $k$, either a graph $\calG_j$ arrives with probability $\frac{\lambda_j}{\xi}$, or a (virtual/actual) template of type $j$ leaves the system with probability $\frac{|C^{(j)}|\mu_j}{\xi}$, or $\ve{S}[k]$ remains unchanged otherwise.
The following proposition states the main ``time-scale decomposition'' property with respect to the associated jump chain (which can be naturally mapped to the original Markov chain).

\begin{proposition}\label{prop-adiabatic}
Let $\nu_n$ denote the (conditional) probability distribution of configuration at index $n$ given the queues $\ve{Q}[n]$ under $\mathrm{DGP}(\tilde{W}(\ve{Q}[n]))$. Let $\pi_n$ be the steady state distribution of configurations corresponding to $\overline{\mathrm{DGP}}(\tilde{W}(\ve{Q}[n]))$. Given any $0<\epsilon<1$, and any initial state $\ve{S}[0]=(\ve{Q}[0],\ve{C}[0])$, there exists a time $n^\star=n^\star(\epsilon,\beta,\ve{S}[0])$ such that for all $n \geq n^\star$, $\|\pi_n-\nu_n\|_{TV}\leq \epsilon/16$.
\end{proposition}

%\begin{proof}
%Let $\tilde{W}^\star(n):=\max_{C \in \calC} \sum_{j \in \calJ} \sum_{A \in C^{(j)}} \tw(n)$. It follows from Corollary \ref{cor1} and definition of $\|\cdot\|_{TV}$ in Proposition~\ref{prop-adiabatic} that
%\ben
%&&\mathbb{E}_{\nu_n}\Big[\sum_{j} \sum_{A \in C^{(j)}} \tw(n)\Big]=\mathbb{E}_{\pi_n}\Big[\sum_{j} \sum_{A \in C^{(j)}} \tw(n)\Big]\\
%&&+\sum_C\Big[(\pi_n(C)-\nu_n(C))\sum_{j} \sum_{A \in C^{(j)}} \tw(n)\Big]\\
%&& \geq \tilde{W}^\star(n)+\beta \min_{C \in \calC} \log \gamma_C -2\delta^\prime \tilde{W}^\star(n)\\
%&& =  (1-2\delta^\prime)\tilde{W}^\star(n)+\beta \min_{C \in \calC} \log \gamma_C
%\een
%\end{proof}
\begin{corollary}\label{cor3}
Given $0<\epsilon<1$, for all $n \geq n^\star (\epsilon, \beta,\ve{S}(0))$,
\begin{align*}
\mathbb{E}_{\nu_n}\Big[\sum_{j} \sum_{A \in C^{(j)}} w^{(j)}_A(n)\Big] \geq \beta \min_{C \in \calC} \log \gamma_C-\epsilon b_{max}\\
+(1-\frac{\epsilon}{4})\max_{C \in \calC} \sum_{j \in \calJ} \sum_{A \in C^{(j)}} w^{(j)}_A(n).
\end{align*}
\end{corollary}
\begin{proof}
Consider any $n \geq n^\star (\epsilon, \beta,\ve{S}(0))$. Let $\tilde{W}^\star(n):=\max_{C \in \calC} \sum_{j \in \calJ} \sum_{A \in C^{(j)}} \tw(n)$. First note that from Corollary \ref{cor1}, Proposition~\ref{prop-adiabatic}, and definition of $\|\cdot\|_{TV}$,
\be
&&\mathbb{E}_{\nu_n}\Big[\sum_{j} \sum_{A \in C^{(j)}} \tw(n)\Big]=\mathbb{E}_{\pi_n}\Big[\sum_{j} \sum_{A \in C^{(j)}} \tw(n)\Big] \nonumber \\
&&+\sum_C\Big[(\pi_n(C)-\nu_n(C))\sum_{j} \sum_{A \in C^{(j)}} \tw(n)\Big]\nonumber \\
&& \geq \tilde{W}^\star(n)+\beta \min_{C \in \calC} \log \gamma_C -2(\frac{\epsilon}{16}) \tilde{W}^\star(n)\nonumber \\
&& =  (1-\frac{\epsilon}{8})\tilde{W}^\star(n)+\beta \min_{C \in \calC} \log \gamma_C. \label{cor2}
\ee
Next, note that by the definition of $w^{(j)}_A$ (see \dref{eq:tildeweight}, \dref{eq:tildeweight2}), for any $j \in \calJ$, $A \in \calA^{(j)}$,
\ben
w^{(j)}_A(n)\leq \tw(n) \leq  w^{(j)}_A(n)+ \frac{\alpha\epsilon}{8M} f(Q_{max}(n)+\bias),
\een
hence for any configuration $C \in \calC$,
\ben \label{ineq1}
0 \leq \sumj \sum_{C \in C^{(j)}} (\tw(n)- w^{(j)}_A(n)) &\leq&  \alpha \frac{\epsilon}{8} f(Q_{max}(n)+\bias).
\een
Suppose $Q^{(j^\prime)}(n)=Q_{\max}(n)$ for some $j^\prime \in \calJ$. Then for any $A^\prime \in \calA^{(j^\prime)}$,
\ben \label{ineq2}
\alpha f(Q_{\max}(n)+\bias)-b^{(j^\prime)}_{A^\prime} &=& {w}^{(j^\prime)}_{A^\prime}(n) \nonumber \\
& \leq & \max_{C \in \calC} \sum_{j \in \calJ} \sum_{A \in C^{(j)}}{w}^{(j)}_{A}(n).
\een
Therefore, it follows that
\ben
0 \leq  \sumj \sum_{C \in C^{(j)}} (\tw(n) -w^{(j)}_A(n)) \leq  \frac{\epsilon}{8}  \max_C \sumj \sum_{C \in C^{(j)}}   w^{(j)}_A\\
 +  \frac{\epsilon}{8} b_{max}.
\een
Let $W^\star(n):=\max_{C \in \calC} \sum_{j \in \calJ} \sum_{A \in C^{(j)}} w^{(j)}_A(n)$. Using the above inequality and \dref{cor2},
\begin{flalign*}
&\mathbb{E}_{\nu_n}\Big[\sum_{j} \sum_{A \in C^{(j)}} w^{(j)}_A(n)\Big] \\
&\geq  \mathbb{E}_{\nu_n}\Big[\sum_{j} \sum_{A \in C^{(j)}} \tw(n)\Big]- \frac{\epsilon}{8}   W^\star(n) -\frac{\epsilon}{8} b_{max}\\
& \geq  (1-\frac{\epsilon}{8})\tilde{W}^\star(n) +\beta \log \gamma_{min}- \frac{\epsilon}{8}   W^\star (n)-\frac{\epsilon}{8} b_{max}\\
& \geq  (1 -\frac{\epsilon}{4} )W^\star (n)+\beta \log \gamma_{min}-\frac{\epsilon}{8} b_{max}.
\end{flalign*}
\end{proof}
\begin{proof}[of Proposition~\ref{prop-adiabatic}]
Below we mention a sketch of the proof of the ``time-scale decomposition'' property for our algorithm.

Let $\Phi^Q$ be the infinitesimal generator of the Markov chain $(C(t))$ under $\overline{\mathrm{DGP}}(\tilde{W}(Q))$, for some vector of queues  $Q$. Let $P^Q_{\xi}=I+\frac{1}{\xi}\Phi^Q$ denote the corresponding transition probability matrix of the jump chain $(C[n])$, obtained by uniformizing $(C(t))$ using the Poison clock $N_\xi(t)$ of rate $\xi$ in \dref{eq:xi}. We use $P^{Q}_\xi (C,C^\prime)$ to denote the transition probability from configuration $C$ to configuration $C^\prime$.

The Markov chain $(C[n])$ is irreducible, aperiodic, and reversible, with the unique steady-state distribution $\pi$ in \dref{expform}. In this case, it is well known that the convergence to the steady-state distribution is geometric with a rate equal to the Second Largest Eigenvalue Modulus (SLEM) of $P^Q_{\xi}$ \cite{bremaud}. Further, using the choice of $\xi$ in \dref{eq:xi}, $(C[n])$ is a lazy Markov chain because at each jump index $n$, the chain will remain in the same state with probability greater than $1/2$. In this case, for any initial probability distribution $\mu_0$ and for all $n \geq 0$,
\begin{equation}\label{mixing1}
\|\mu_0 ({P}^Q_{\xi})^n-\pi\| _{TV} \leq \theta_2^n \frac{1}{2\sqrt{\pi_{min}}},
\end{equation}
where $\theta_2$ is the second largest eigenvalue of $P^Q_\xi$, and $\pi_{min}=\min_C \pi(C)$. Correspondingly, the mixing time of the chain (defined as $\inf \{n>0: \|\nu(n)-\pi(n)\|_{TV} \le \delta\}$) will be less than $\frac{-\log(2 \delta \sqrt{\pi_{min}})}{ (1-\theta_2)}$.

Lemma \ref{lemma:slem} below provides a bound on $\theta_2$ and hence on the convergence rate of Markov chain $P^Q_{\xi}$.
\begin{lemma}\label{lemma:slem}
Let $K_0=\left(\frac{\rho_{min}\wedge 1}{\rho_{max}\vee 1} \right)^M\frac{\wedge_j (\mu_j \wedge \lambda_j)}{|\calC|(M!)^2}$. Then,
\be \label{eq:slem}
\frac{1}{1-\theta_2}\leq \frac{2 \xi^2}{K_0^2}\exp\Big[\frac{2(M+1)}{\beta}( \alpha f(Q_{max}+\bias)+b_{max})\Big].
\ee
\end{lemma}
\begin{proof}[of Lemma~\ref{lemma:slem}]
It follows from Cheeger's inequality~\cite{bremaud} that $\frac{1}{1-\theta_2}\leq \frac{2}{\Psi^2(P_\xi)}$ where $\Psi(P_\xi)$ is the conductance of the Markov chain $P^Q_{\xi}$. The conductance is further bounded from below as
\be
\Psi(P_\xi)&\ge& 2 \pi_{min}\min_{C \neq C^\prime}P^{Q}_\xi (C,C^\prime).
\ee
Under $\overline{\mathrm{DGP}}(\tilde{W})$, with $\tilde{W}=\tilde{W}(Q)$,
\ben
\min_{C \neq C^\prime}P^Q_\xi (C,C^\prime) &=& \frac{1}{\xi}\min_j \frac{\mu_j}{1+e^{\frac{1}{\beta}\tw}}\wedge \frac{\lambda_j}{|\calA^{(j)}(C)|}\frac{e^{\frac{1}{\beta}\tw}}{1+e^{\frac{1}{\beta}\tw}}\\
& \geq & \frac{\wedge_j (\mu_j \wedge \lambda_j) }{\xi M !}\frac{\exp(\frac{1}{\beta}\tilde{w}_{min})\wedge 1}{1+\exp(\frac{1}{\beta}\tilde{w}_{max})}\\
& \geq & \frac{\wedge_j (\mu_j \wedge \lambda_j) }{\xi M!}\frac{\exp(\frac{-1}{\beta}b_{max})}{1+\exp(\frac{\alpha}{\beta}f(Q_{max}+\bias))}
\een
Note that the steady state distribution of the jump chain is still $\pi(C)=\frac{\gamma(C)}{Z_\beta}\exp(\frac{1}{\beta}\sum_j\sum_{A\in C} \tw)$, for $\gamma$ defined in \dref{gammaform}.
Then
\ben
Z_\gamma Z_\beta &\leq & \sum_{C \in \calC} \exp(\frac{1}{\beta}\sum_j\sum_{A\in C} \tilde{w}_{max}) M! (\rho_{max} \vee 1)^M\\
& \leq & |\calC|\exp(\frac{M\alpha}{\beta} f(Q_{max}+\bias)) M! (\rho_{max} \vee 1)^M,
\een
therefore
\be \label{pi-min}
\pi_{min} \geq K_1 \exp\Big(-\frac{M\alpha}{\beta} f(Q_{max}+\bias)-\frac{M}{\beta}b_{max}\Big),
\ee
where $K_1=\left(\frac{\rho_{min}\wedge 1}{\rho_{max}\vee 1} \right)^M\frac{1}{|\calC|M!}$.
Hence
\ben
\Psi(P^Q_\xi) \geq \frac{K_0}{\xi}\exp\Big(-\frac{M+1}{\beta}( \alpha f(Q_{max}+\bias)+b_{max})\Big).
\een
where $K_0=K_1\frac{\wedge_j (\mu_j \wedge \lambda_j)}{M!}$.
\end{proof}
\begin{lemma} \label{alpha ratio}
For any configuration $C \in \calC$,
$
e^{-\sigma_n}\leq \frac{\pi_{n+1}(C)}{\pi_n(C)} \leq e^{\sigma_n},
$
where
\begin{equation} \label{eq: alpha}
\sigma_n= \frac{2M\alpha}{\beta} f^\prime\Big(f^{-1}\Big(\frac{\epsilon}{8M} f(h+Q_{max}(n+1))\Big)-1\Big).
\end{equation}
\end{lemma}
\begin{proof}[of Lemma \ref{alpha ratio}]
Note that
$$
\frac{\pi_{n+1}(C)}{\pi_n(C)}=\frac{Z_n(\beta)}{Z_{n+1(\beta)}}e^{\frac{\alpha}{\beta}\sum_{j,A \in C^{(j)}}f^{(j)}(Q(n+1)+h)-f^{(j)}(Q(n)+h)}.
$$
%or
%$$
%\exp\left(-\sum_{i \in \rho}|\widetilde{w}_i(t+1)-\widetilde{w}_i(t)|\right) \leq \frac{\pi_{t+1}(\rho)}{\pi_t(\rho)} \leq \exp\left(\sum_{i \in \rho}|\widetilde{w}_i(t+1)-\widetilde{w}_i(t)|\right).
%$$
It is easy to show that
\ben
\frac{Z_n(\beta)}{Z_{n+1}(\beta)}  \leq  \max_C e^{\frac{\alpha}{\beta}\sum_{j,A \in C^{(j)}}f^{(j)}(Q(n+1)+h)-f^{(j)}(Q(n)+h)}.
\een
Let $Q^*(n):=f^{-1}(\frac{\epsilon}{8M} f(h+Q_{max}(n)))-h$, and define $\tilde{Q}^{(j)}(n):=\max\{Q^*(n), Q^{(j)}(n)\}$.
Then,
\begin{flalign*}
&f^{(j)}(Q(n+1)+h)-f^{(j)}(Q(n)+h)\\
&  =f(\tilde{Q}^{(j)}(n+1)+h)-f(\tilde{Q}^{(j)}(n)+h) \\
& \leq f^\prime(\tilde{Q}^{(j)}(n+1)+h-1)|\tilde{Q}^{(j)}(n+1)-\tilde{Q}^{(j)}(n)|  \\
& \leq f^\prime({Q}^{\star}(n+1)+h-1)\\
&  = f^\prime(f^{-1}(\frac{\epsilon}{8M} f(h+Q_{max}(n+1)))-1)
\end{flalign*}
 where we have used the mean value theorem and the facts that $f$ is a concave increasing function and at each index $n$, one queue can change at most by one.
Therefore,
\ben
\frac{\pi_{n+1}(C)}{\pi_n(C)} \leq e^{2\frac{M\alpha}{\beta} f^\prime(f^{-1}(\frac{\epsilon}{8M} f(h+Q_{max}(n+1)))-1)}.
\een
A similar calculation shows that also
\ben
\frac{\pi_{n}(C)}{\pi_{n+1}(C)} \leq e^{2\frac{M\alpha}{\beta} f^\prime(f^{-1}(\frac{\epsilon}{8M} f(h+Q_{max}(n+1)))-1)}.
\een
This concludes the proof.
\end{proof}
Next, we use the following version of Adiabatic Theorem from~\cite{RSS09} to prove the time-scale decomposition property of our algorithm.
\begin{proposition}\label{drift}
(Adapted from \cite{RSS09}) Suppose
\begin{equation}\label{alphaT}
\frac{\sigma_n}{1-\theta_2(n+1)} \leq \delta^\prime/4 \mbox{ for all } n \geq 0,
\end{equation}
for some $\delta^\prime >0$, where $\theta_2(n+1)$ denotes the second largest eigenvalue of $P^{Q(n+1)}_\xi$. Then $\|\pi_n-\nu_n\|_{TV} \leq \delta^\prime$, for all $n \geq n^\star(\beta,\delta^\prime,\ve{S}(0))$, where
$n^\star$ is the smallest $n$ such that
\be\label{n*}
\frac{1}{\sqrt{\pi_{min}(0)}} \exp(-\sum_{k=0}^{n}(1-\theta_2(k))^2 \leq \delta^\prime.
\ee
%\begin{equation}\label{t*}
%\sum_{k=t_1:\|q(t_1)\|=q_{th}}^{t_1+t^*}\frac{1}{T^2_k} \geq \log (4/\delta)+|\mathcal{L}|(g(q_{th})+\log2)/2.
%\end{equation}
\end{proposition}
%In the above Lemma, condition (ii) is based on the upper bound of (\ref{mixing3}) and the fact that $\tilde{w}_{max}(t)\leq (1+\frac{\epsilon}{4 |\mathcal{L}|})w_{max}(t)$.
%See the appendix for the proof of the above Lemma.
In our context, Proposition \ref{drift} states that under
\dref{alphaT} and \dref{n*}, after $n^\star$ steps, the distribution
of the configurations over templates will be close to the desired
steady-state distribution. To get some intuition, $\sigma_n$ has the
interpretation of the rate at which weights change, and
$1/(1-\theta_2(n+1))$ has the interpretation of the time taken for the
system to reach steady-state after the weights change. Thus,
condition~\dref{alphaT} ensures a {\em time-scale decomposition} --
the weights change slowly compared to the time-scale that the system
takes in order to respond and ``settle down'' with these changed weights.

It remains to show that that our system indeed satisfies the conditions of Proposition~\ref{drift} as we do next, for the choice of $\delta^\prime =\frac{\epsilon}{16}$. Suppose $f(x)=\log^{1-b}(x)$, for some $0 <b<1$. Let $y=f(Q_{max}(n+1)+h)$. Obviously $f^\prime(x)\leq 1/x$, so in view of equations \dref{eq: alpha}, \dref{eq:slem}, \dref{alphaT}, it suffices to have
\ben
\frac{2M\alpha}{\beta}\frac{1}{f^{-1}(\frac{\epsilon}{8M} y)-1}\exp\left[\frac{4M}{\beta}(b_{max}+\alpha y)\right] \leq \frac{K_0^2 \epsilon}{128\xi^2}.
\een
Note that $f^{-1}(x)=\exp(x^{\frac{1}{1-b}})$. Suppose $\alpha \leq \beta$. A simple calculation shows that
%and
%$
%y \geq \frac{8M}{\epsilon}\log^{1-b} 3,
%$
%then $\frac{\exp(-(\frac{\epsilon}{8M} y) ^{\frac{1}{1-b}})}{1-\exp(-(\frac{\epsilon}{8M} y) ^{\frac{1}{1-b}})}\leq 0.5$, and it suffices that
%\ben
%2M\exp\left[\frac{4M}{\beta}b_{max}+4M\mu_{max}y-(\frac{\epsilon}{8M} y)^{\frac{1}{1-b}}\right] \leq \frac{K_0^2 \epsilon}{256\xi^2}
%\een
%or equivalently
%\ben
%\frac{4M}{\beta}b_{max}+4M\mu_{max}y-(\frac{\epsilon}{8M} y)^{\frac{1}{1-b}} \leq \log\left(\frac{K_0^2 \epsilon}{512\xi^2M}\right).
%\een
it suffices to jointly have
\begin{flalign*}
&y \geq \frac{8M}{\epsilon}\log^{1-b} 3,\\
&4M\mu_{max}y-\frac{1}{2}(\frac{\epsilon}{8M} y)^{\frac{1}{1-b}}\leq 0,\\
&\frac{4M}{\beta}b_{max}-\frac{1}{2}(\frac{\epsilon}{8M} y)^{\frac{1}{1-b}}\leq \log\left(\frac{K_0^2 \epsilon}{512\xi^2M}\right).
\end{flalign*}
In summary, the condition \dref{alphaT} holds if
%\ben
%y \geq \max\Big\{\Big(\frac{1}{\epsilon}\Big)^{1/b}(8M)^{2/b}\mu_{\max}^{(1-b)/b},\frac{8M}{\epsilon}\log^{1-b} 3\\
%\frac{8M}{\epsilon}\Big(\frac{8M}{\beta}b_{max}+2 \log(\frac{32\xi^2M}{K_0^2 \delta^\prime})\Big)^{1-b}\Big\}
%\een
%Now in our algorithm, $y\geq \log^{1-b}(h)=\frac{1}{\beta}$, which satisfies the condition as $\beta \to 0$.
\be
y \geq \Big(\frac{1}{\epsilon}\Big)^{2-b+\frac{1}{b}} \Big(\frac{1}{\beta}\Big)^{1-b} C_0
\ee
or as a sufficient condition, if
\be \label{eq:hcond}
h\geq \exp \Big(C_0 \frac{1}{\beta} (\frac{1}{\epsilon})^{\frac{2-b+\frac{1}{b}}{1-b}} \Big)
\ee
for
$$C_0\geq 8M\Big(8Mb_{\max}+2|\log \frac{512 \xi^2}{K_0^2}|+2+(8M)^{2/b}\mu_{\max}^{(1-b)/b}\Big).$$

Next, we find $n^\star$ that satisfies
\ben
\sum_{k=0}^{n^\star-1}(1-\theta_2(k))^2 \geq -\log(\frac{\epsilon}{16})-\frac{1}{2}\log (\pi_{min}(0)).
\een
From (\ref{pi-min}), and since $\alpha \leq \beta$,
\ben
-\log(\pi_{min}(0)) \leq \log K_1+M \mu_{max}f(Q_{max}(0)+h)+\frac{M}{\beta}b_{max}.
\een
%where $K_1=\left(\frac{\rho_{min}\wedge 1}{\rho_{max}\vee 1} \right)^M\frac{1}{|\calC|M!}$.
Using Lemma~\ref{lemma:slem}, it can be shown that
\begin{flalign*}
&\sum_{k=0}^{n^\star-1}(1-\theta_2(k))^2   \\
&\geq \frac{2 \xi^2}{K_0^2} e^{-4\frac{Mb_{max}}{\beta}}\sum_{k=0}^{n^\star-1}e^{-4M\mu_{max}f(Q_{max}(k)+h)}\\
& \geq \frac{2 \xi^2}{K_0^2} e^{-4\frac{Mb_{max}}{\beta}}\sum_{k=0}^{n^\star-1}e^{-4M\mu_{max}f(Q_{max}(0)+\bias+n^\star)}\\
%& \geq  \frac{2 \xi^2}{K_0^2} e^{-4\frac{Mb_{max}}{\beta}}n^\star e^{-4M\mu_{max}f(Q_{max}(0)+\bias+n^\star)}\\
%&\geq  \frac{2 \xi^2}{K_0^2} e^{-4\frac{Mb_{max}}{\beta}}n^\star (Q_{max}(0)+h+n^\star)^{\frac{-4M\mu_{max}}{\log^b(Q_{max}(0)+h+n^\star)}}\\
& \geq \frac{2 \xi^2}{K_0^2} e^{-4\frac{Mb_{max}}{\beta}}n^\star (Q_{max}(0)+h+n^\star)^{\frac{-4 M\mu_{max}}{\log^b h}}.
\end{flalign*}
For $h \geq \exp((8M\mu_{max})^{1/b})$, it then suffices that
\begin{align*}
&\frac{2 \xi^2}{K_0^2} e^{-4\frac{Mb_{max}}{\beta}}n^\star (Q_{max}(0)+h+n^\star)^{-1/2} \geq\\
&  \log(\frac{16}{\epsilon})+ \frac{1}{2}\log K_1+\frac{M \mu_{max}}{2}f(Q_{max}(0)+h)+\frac{M}{2\beta}b_{max}
\end{align*}
%Choose $n^\star=Q_{max}(0)+h$, then it suffices that
%\begin{align*}
%&\frac{ \xi^2}{K_0^2} e^{-4\frac{Mb_{max}}{\beta}}\sqrt{h}\geq \log(\frac{16}{\epsilon})+ \frac{1}{2}\log K_1\\
%&+\frac{M \mu_{max}}{2}(\log Q_{max}(0)+\log h)+\frac{M}{2\beta}b_{max},
%\end{align*}
which is clearly satisfied by choosing $n^\star=Q_{max}(0)+h$ for $h$ in \dref{eq:hcond} and $C_0$ a large enough constant.
%Recall that $h=\exp\Big((\frac{1}{\beta})^{1/(1-b)}\Big)$, and hence the inequality
%%Suppose $\beta \leq \frac{1}{10Mb_{max}}$, then it suffices that
%%\ben
%%\frac{ \xi^2}{K_0^2} e^{\frac{Mb_{max}}{\beta}}\geq \log(\frac{4}{\delta^\prime})+ \frac{1}{2}\log K_1+\frac{M \mu_{max}}{2}(\log Q_{max}(0)+1/{\beta^2})+\frac{M}{2\beta}b_{max}
%%\een
%is clearly satisfied for $\beta$ small enough.
\end{proof}
\subsubsection*{Step 3: Lyapunov Analysis}
The final step of the proof is based on a Lyapunov optimization method \cite{neely}.
We develop the required Lyapunov arguments for $S(k)=(Q(k),C(k))$, i.e., the jump chain of the uniformized Makov chain.
Consider the following Lyapunov function $$V(k)=\sum_{j\in \calJ} \frac{1}{\mu_j}F(Q^{(j)}(k)+\bias),$$ where
$F(x)= \int_h^{x} f(\tau)\dd \tau$. Recall that $f(x)=\log^{1-b} x$. Therefore $F$ is convex, and following the standard one-step drift analysis
\begin{flalign*}
&V(k+1)-V(k) \leq \\
& \sumj \frac{1}{\mu_j}f(Q^{(j)}(k+1)+\bias)\Big(Q^{(j)}(k+1)-Q^{(j)}(k)\Big)=\\
& \sumj \frac{1}{\mu_j}f(Q^{(j)}(k)+\bias)\Big(Q^{(j)}(k+1)-Q^{(j)}(k)\Big)+  \\
& \sumj \frac{1}{\mu_j}\Big(f(Q^{(j)}(k+1)+\bias)-f(Q^{(j)}(k)+\bias)\Big)\Big(Q^{(j)}(k+1)\\
&-Q^{(j)}(k)\Big).
\end{flalign*}
By the mean value theorem, and using the fact that $f$ is a concave increasing function, it follows that
\begin{align*}
&|f(Q^{(j)}(k+1)+\bias)-f(Q^{(j)}(k)+\bias)| \leq \\
%& |f^\prime(\varsigma)||Q^{(j)}(k+1)-Q^{(j)}(k))|\leq \\
&  f^\prime(\bias)|Q^{(j)}(k+1)-Q^{(j)}(k))|
\end{align*}

Recall that $C(k)=\left(C_a^{(j)}(k),C_v^{(j)}(k)\right)$ where $C_v^{(j)}$ is the set of virtual templates (i.e., the templates that do not contain jobs of type $j$) and $C_a^{(j)}$ is the set of actual templates.

For notational compactness, let $\mathbb{E}_{\ve{S}(k)}[\cdot]=\mathbb{E}[\cdot|\ve{S}(k)]$, where $\ve{S}(k)$ is the state of the system at each index $k$. Then
\begin{align*}
&\expectsk{V(k+1)-V(k)} \leq f^\prime(h) \sum_j \frac{1}{\mu_j} \Big(\frac{\lambda_j}{\xi}+\frac{|C^{(j)}(k)\mu_j}{\xi}\Big) +\\
&\sum_j  \frac{1}{\mu_j} {f(Q^{(j)}(k)+\bias)\Big[\frac{\lambda_j}{\xi}-\left(|C^{(j)}(k)|-|C_v^{(j)}(k)|\right)\frac{\mu_j}{\xi} \Big]},
\end{align*}
where we have used the fact that at most one arrival or departure can happen at every jump index, i.e., $|Q^{(j)}(k+1)-Q^{(j)}(k)|\in \{0,1\}$.

%Let $\mu_{\max}:=\max_j \mu_j$. Note that clearly
%$N_{max}=\lfloor \suml |m_\ell| / \min_j |V_j|\rfloor$
Note that clearly the maximum number of templates of any type of jobs that can fit in a configuration is less than $M$ (recall that $M=\sum_{\ell \in \calL} |m_\ell|$). Moreover, none of the templates of type $j$ will be virtual if more than $M$ jobs of type $j$ are available in the system, hence, 
\begin{align*}
&|C_v^{(j)}(k)| f(\bias+Q^{(j)}(k))\leq |C_v^{(j)}(k)| f(\bias+M).
%&|C_v^{(j)}(t)| f(\bias+Q^{(j)}(t)) \mathds{1}{\{Q^{(j)}(t)\leq M\}} \leq |C_v^{(j)}(t)| f(\bias+M).
\end{align*}
%\be
%f(Q^{(j)}(k)+\bias)|C_v^{(j)}(k) \leq f(N_{max}+\bias)|N_{max}
%\ee
and therefore,
\begin{flalign*}
& \expectsk{V(k+1)-V(k)}  \leq K_2+K_3 \\
&+ \frac{1}{\xi} \sum_j  {f(Q^{(j)}(k)+\bias)\Big({\rho_j}-|C^{(j)}(k)|\Big)}
\end{flalign*}
where $K_2\leq f^\prime(\bias)(M+\sum_j \rho_j)/\xi$, and $K_3 \leq Mf(\bias+M)/{\xi}$. Therefore, it follows that
\begin{align*}
&\alpha \expectsk{V(k+1)-V(k)}+\frac{1}{\xi}\expectsk{\sumj\sumA \bjA} \leq \\
& \alpha (K_2+ K_3) + \frac{\alpha}{\xi} \sumj  \rho_j f\Big(Q^{(j)}(k)+\bias\Big)\\
&-\frac{1}{\xi}\sumj\sumA \Big[\alpha f\Big(Q^{(j)}(k)+\bias\Big)-\bjA \Big]
\end{align*}
Taking the expectation of both sides with respect to $\nu_n$ (distribution of configurations given the queues at $n>n^*$), we get
\begin{align}
&\alpha \expectqk{V(k+1)-V(k)} +\frac{1}{\xi}\expectqk{\sumj \sum_{A\in \calA^{(j)}} x_A(k) \bjA}  \nonumber \\
& \leq \alpha(K_2+ K_3)+ \frac{\alpha}{\xi} \sumj  \rho_j f(Q^{(j)}(k)+\bias)\nonumber\\
&-\frac{1}{\xi}\expectqk{\sumj\sumA (\alpha f(Q^{(j)}(k)+\bias)-\bjA )}\nonumber \\
&  \leq \alpha( K_2+K_3) -\frac{\beta}{\xi}\log \gamma_{min}+\frac{\alpha}{\xi}\sumj  \rho_j f(Q^{(j)}(k)+\bias)\nonumber \\
&-\frac{1}{\xi}(1-\frac{\epsilon}{4})W^\star(k)+\frac{\epsilon}{\xi} b_{\max},\label{eq:deltaplus}
\end{align}
where the last inequality is based on Corollary \ref{cor3}, where
\ben
W^\star(k)=\max_{C \in \calC} \sumj\sum_{A\in C^{(j)}} \Big(\alpha f(Q^{(j)}(k)+\bias)-\bjA \Big)
\een
Notice that equivalently
\ben
W^\star(k)&=\displaystyle{\max_{\{x_A\}}} &\sumj\sum_{A\in \calA^{(j)}} x_A\Big(\alpha f(Q^{(j)}(k)+\bias)-\bjA \Big)\\
 &\mbox{subject to} &(\sum_{A \in \calA^{(j)}}x_A; j \in \calJ) \in \Lambda\\
 &&x_A \geq 0; \forall A \in \cup_j \calA^{(j)}
\een
Let $x^\star$ be the optimal solution to the static partitioning problem. By the feasibility of $x^\star$, $\rho_j \leq \sum_{A\in \calA^{(j)}}x^\star_A$, for all $j \in \calJ$, hence
\be \label{eq:term1}
\sumj  \rho_j f(Q^{(j)}(k)+\bias) \leq \sumj \sum_{A \in A^{(j)}} {{x_A^\star}} f(Q^{(j)}(k)+\bias).
\ee
Further, by assumption, $\rho$ is strictly inside $\Lambda$, thus there exists a $\delta^\star$ such that $\rho(1+\delta^\star) \in \Lambda$. It is easy to show by the monotonicity of $\calC$ (i.e., if $C\in \calC$, $C \backslash A \in \calC$, for all $A \in C^{(j)}$, $j \in \calJ$) that, at the optimal solution, the constraint \dref{cons1} should in fact hold with equality. Hence 
$$\left(\sum_{A \in \calA^{(j)}}{x^\star}_{A}(1+\delta^\star ): j \in \calJ \right) \in \Lambda.$$
%for any $0 \leq \delta \leq \delta^\star$.
Therefore,
\be \label{eq:term2}
 W^\star(k)\geq \sumj \sum_{A \in \calA^{(j)}} \left((1+\delta){x^\star}_{A}\right) \Big(\alpha f(Q^{(j)}(k)+\bias)-b^{(j)}_{A}  \Big)
\ee
For $\epsilon \leq \delta^\star$, $(1-\frac{\epsilon}{4})(1+\delta^\star) \geq 1+\frac{\delta}{2}$, for any $\delta \in [0,\delta^\star].$ Then using \dref{eq:term1} and \dref{eq:term2} in \dref{eq:deltaplus},
\begin{align}
&\alpha \expectqk{V(k+1)-V(k)} +\frac{1}{\xi}\expectqk{\sumj \sum_{A\in \calA^{(j)}} x_A(k) \bjA}  \nonumber \\
& \leq \alpha (K_2+ K_3)+\frac{\alpha}{\xi} \sumj \sum_{A \in A^{(j)}} {{x_A^\star}} f(Q^{(j)}(k)+\bias) \nonumber\\
& - \frac{1}{\xi}(1+\frac{\delta}{2})\sumj \sum_{A \in \calA^{(j)}} {x^\star}_{A} \Big(-b^{(j)}_{A} +\alpha f(Q^{(j)}(k)+\bias)  \Big) \nonumber\\
&-\frac{\beta}{\xi} \log \gamma_{min} +\frac{\epsilon}{\xi} b_{\max}\nonumber\\
%&  \leq  \frac{\alpha}{\xi} \sumj \sum_{A \in A^{(j)}}\mu_j {{x_A^\star}} f(Q^{(j)}(t))- \frac{1}{\xi}(1+{\delta}{2})\sumj \sum_{A \in \calA^{(j)}} {x^\star}_{A} \left(-b^{(j)}_{A}+\alpha f(Q^{(j)}(t)+\bias)  \right) \\
%& -\frac{\beta}{\xi} \log \gamma_{min}+\alpha K_1+\alpha K_2 = \\
&=\frac{1}{\xi}(1+\frac{\delta}{2}) G(x^\star) -\frac{\alpha}{\xi}\frac{\delta}{2}\sumj \sumA x_A^\star f(Q^{(j)}(k)+\bias)\nonumber\\
&-\frac{\beta}{\xi} \log \gamma_{min}+\frac{\epsilon}{\xi} b_{\max}+\alpha (K_2+ K_3) \label{eq:deltaplus2}.
\end{align}
It follows from this that the Markov chain $(\ve{Q}(k),\ve{C}(k))$ is positive recurrent as a consequence of the Foster-Lyapunov theorem, with Lyapunov function $V(\cdot)$. Taking the expectation of both sides of \dref{eq:deltaplus2} with respect to $\ve{Q}(k)$, and then taking summation over $k=0,..., N-1$, and dividing by $N$, and letting $N\to \infty$ yields
\begin{flalign*}
&\limsup_N \frac{1}{N}\sum_{k=0}^{N-1}\sumj\expect{f(Q^{(j)}(k)+\bias)} \\
&\leq \frac{\alpha \xi (K_2+K_3)-\beta \log\gamma_{min}+\epsilon b_{\max}+(1+\delta/2)G(x^\star)}{\alpha \rho_{min}{\delta/2}}\\
&\limsup_N \frac{1}{N}\sum_{k=0}^{N-1}\expect{G(x(k))}\\
&\leq (1+\delta/2)G(x^\star)+\alpha \xi (K_2 + K_3)-\beta \log\gamma_{min}+\epsilon b_{\max}
\end{flalign*}
where we have used the fact that contribution of queue sizes and costs in $(0,n^\star]$ to the average quantities vanishes to zero as $N\to \infty$.
The above inequalities can be independently optimized over $\delta \in [0, \delta^\star]$ (the performance of the algorithm is independent of $\delta$). Here we choose $\delta=\delta^*$ in the queue inequality and $\delta=0$ in the cost inequality.  The statement of Theorem then follows using the Ergodic theorem and the fact that the jump chain and the original chain have the same steady-state average behaviour.

\section{Conclusions}\label{sec:con}

Motivated by modern stream data processing applications, we have
investigated the problem of dynamically scheduling graphs in cloud
clusters, where a graph represents a specific computing job. These
graphs arrive and depart over time. Upon arrival, each graph can
either be queued or served immediately. The objective is to develop
algorithms that assign nodes of these graphs to free (computing)
slots in the machines of the cloud cluster. The performance metric for
the scheduler (partition graphs and map it to slots) is to minimize
the average graph partitioning cost, while keeping the system stable.

We have proposed a novel class of low complexity algorithms which can
approach the optimal solution by exploiting the trade-off between
delay and partitioning cost, without causing service interruptions.
The key ingredient of the algorithms is the generation/removal of
random templates from the cluster at appropriate instances of time,
where each template is a unique way of partitioning a graph.

%\balance
\bibliographystyle{abbrv}
\bibliography{storm.bbl}

\begin{thebibliography}{10}

\bibitem{andreev2006balanced}
K.~Andreev and H.~Racke.
\newblock Balanced graph partitioning.
\newblock {\em Theory of Computing Systems}, 39(6):929--939, 2006.

\bibitem{LRL13}
L.~Aniello, R.~Baldoni, and L.~Querzoni.
\newblock Adaptive online scheduling in {Storm}.
\newblock In {\em 7th ACM international conference on Distributed event-based
  systems}, pages 207--218, 2013.

\bibitem{bremaud}
P.~Bremaud.
\newblock {\em Markov chains: Gibbs fields, Monte Carlo simulation, and
  queues}, volume~31.
\newblock springer, 1999.

\bibitem{graphpartition13}
A.~Bulu{\c{c}}, H.~Meyerhenke, I.~Safro, P.~Sanders, and C.~Schulz.
\newblock Recent advances in graph partitioning.
\newblock {\em CoRR}, arXiv:1311.3144, 2013.

\bibitem{survey1}
E.~G. Coffman~Jr, M.~R. Garey, and D.~S. Johnson.
\newblock Approximation algorithms for bin packing: A survey.
\newblock In {\em Approximation algorithms for NP-hard problems}, pages 46--93.
  PWS Publishing Co., 1996.

\bibitem{onlinemapreduce10}
T.~Condie, N.~Conway, P.~Alvaro, J.~M. Hellerstein, K.~Elmeleegy, and R.~Sears.
\newblock {MapReduce} online.
\newblock In {\em NSDI}, volume~10, page~20, 2010.

\bibitem{survey2}
J.~Csirik and G.~J. Woeginger.
\newblock {\em On-line packing and covering problems}.
\newblock Springer, 1998.

\bibitem{feldmann2012balanced}
A.~E. Feldmann and L.~Foschini.
\newblock Balanced partitions of trees and applications.
\newblock {\em Algorithmica}, pages 1--23, 2012.

\bibitem{ghaderiflow}
J.~Ghaderi, T.~Ji, and R.~Srikant.
\newblock Flow-level stability of wireless networks: Separation of congestion
  control and scheduling.
\newblock {\em IEEE Transactions on Automatic Control}, 59(8):2052 -- 2067,
  2014.

\bibitem{GS10}
J.~Ghaderi and R.~Srikant.
\newblock On the design of efficient {CSMA} algorithms for wireless networks.
\newblock In {\em 49th IEEE Conference on Decision and Control (CDC)}, pages
  954--959, 2010.

\bibitem{GZS14}
J.~Ghaderi, Y.~Zhong, and R.~Srikant.
\newblock Asymptotic optimality of {BestFit} for stochastic bin packing.
\newblock {\em SIGMETRICS Perform. Eval. Rev.}, 42(2):64--66, Sept. 2014.

\bibitem{giraph}
Giraph.
\newblock \url{https://giraph.apache.org/}.

\bibitem{graphlab}
Graphlab.
\newblock \url{http://graphlab.com/}.

\bibitem{hendrickson1995multi}
B.~Hendrickson and R.~W. Leland.
\newblock A multi-level algorithm for partitioning graphs.
\newblock {\em SC}, 95:28, 1995.

\bibitem{sphere}
IBM InfoSphere Platform, \url{http://www-01.ibm.com/software/data/infosphere}.

\bibitem{chiang}
J.~W. Jiang, T.~Lan, S.~Ha, M.~Chen, and M.~Chiang.
\newblock Joint {VM} placement and routing for data center traffic engineering.
\newblock In {\em Proceedings of IEEE INFOCOM}, pages 2876--2880, 2012.

\bibitem{jiang2010distributed}
L.~Jiang and J.~Walrand.
\newblock A distributed {CSMA} algorithm for throughput and utility
  maximization in wireless networks.
\newblock {\em IEEE/ACM Transactions on Networking (TON)}, 18(3):960--972,
  2010.

\bibitem{lippman75}
S.~A. Lippman.
\newblock Applying a new device in the optimization of exponential queuing
  systems.
\newblock {\em Operations Research}, 23(4):687--710, 1975.

\bibitem{liu2010towards}
J.~Liu, Y.~Yi, A.~Proutiere, M.~Chiang, and H.~V. Poor.
\newblock Towards utility-optimal random access without message passing.
\newblock {\em Wireless Communications and Mobile Computing}, 10(1):115--128,
  2010.

\bibitem{neely}
M.~J. Neely.
\newblock {\em Stochastic network optimization with application to
  communication and queueing systems}, volume~3.
\newblock Morgan \& Claypool Publishers, 2010.

\bibitem{puterman09}
M.~L. Puterman.
\newblock {\em Markov decision processes: discrete stochastic dynamic
  programming}, volume 414.
\newblock John Wiley \& Sons, 2009.

\bibitem{QH13}
Z.~Qian, Y.~He, C.~Su, Z.~Wu, H.~Zhu, T.~Zhang, L.~Zhou, Y.~Yu, and Z.~Zhang.
\newblock {TimeStream: Reliable stream computation in the cloud}.
\newblock In {\em EuroSys 2013}, pages 1--14, 2013.

\bibitem{RSS09}
S.~Rajagopalan, D.~Shah, and J.~Shin.
\newblock Network adiabatic theorem: An efficient randomized protocol for
  contention resolution.
\newblock In {\em ACM SIGMETRICS Performance Evaluation Review}, volume~37,
  pages 133--144. ACM, 2009.

\bibitem{rychlyscheduling}
M.~Rychly, P.~Koda, and P.~Smrz.
\newblock Scheduling decisions in stream processing on heterogeneous clusters.
\newblock In {\em Eighth International Conference on Complex, Intelligent and
  Software Intensive Systems (CISIS)}, pages 614 -- 619, July 2014.

\bibitem{s4}
S4 distributed stream computing platform, \url{http://incubator.apache.org/s4}.

\bibitem{shah2012randomized}
D.~Shah and J.~Shin.
\newblock Randomized scheduling algorithm for queueing networks.
\newblock {\em The Annals of Applied Probability}, 22(1):128--171, 2012.

\bibitem{SZ13b}
A.~Stolyar and Y.~Zhong.
\newblock Asymptotic optimality of a greedy randomized algorithm in a
  large-scale service system with general packing constraints.
\newblock {\em arXiv preprint arXiv:1306.4991}, 2013.

\bibitem{S13}
A.~L. Stolyar.
\newblock An infinite server system with general packing constraints.
\newblock {\em Operations Research}, 61(5):1200--1217, 2013.

\bibitem{storm}
Storm: {D}istributed and fault-tolerant realtime computation,
  \url{http://storm.incubator.apache.org}.

\bibitem{walshaw2000mesh}
C.~Walshaw and M.~Cross.
\newblock Mesh partitioning: a multilevel balancing and refinement algorithm.
\newblock {\em SIAM Journal on Scientific Computing}, 22(1):63--80, 2000.

\bibitem{wiki-graph}
Graph partition problem on Wikipedia,
  \url{http://en.wikipedia.org/wiki/Graph_partition}.

\bibitem{ZSS12}
M.~Zaharia, T.~Das, H.~Li, S.~Shenker, and I.~Stoica.
\newblock Discretized streams: an efficient and fault-tolerant model for stream
  processing on large clusters.
\newblock In {\em Proceedings of the 4th USENIX conference on Hot Topics in
  Cloud Ccomputing}, pages 10--10. USENIX Association, 2012.

\end{thebibliography}

%\clearpage
\appendix
\section{Proof of Theorem 1}
The proof is standard and based on a Lyapunov optimization method \cite{neely}.
Consider a Lyapunov function $V(t)=\sum_{j\in \calJ} \frac{1}{\mu_j}F(Q^{(j)}(t))$, where $F(x)= \int_0^ {x} f(\tau)\dd \tau$. Recall that $f: \mathds{R}_+ \to \mathds{R}_+$ is a concave increasing function; thus $F$ is convex. Choose an arbitrarily small $u>0$. It follows from convexity of $F$ that for any $t \geq 0$,
\begin{align*}
&V(t+u)-V(t) \leq \sumj \frac{1}{\mu_j}f(Q^{(j)}(t+u))(Q^{(j)}(t+u)-Q^{(j)}(t))\\
& = \sumj \frac{1}{\mu_j}f(Q^{(j)}(t))(Q^{(j)}(t+u)-Q^{(j)}(t))\\
& + \sumj \frac{1}{\mu_j}(f(Q^{(j)}(t+u))-f(Q^{(j)}(t)))(Q^{(j)}(t+u)-Q^{(j)}(t)).
\end{align*}
By definition $$ Q^{(j)}(t+u)-Q^{(j)}(t)= H^{(j)}(t,t+u)-D^{(j)}(t,t+u),$$ where for any $0 \leq t_1 \leq t_2$,
\ben
H^{(j)}(t_1,t_2)=\calN_j^{H}(\lambda_j(t_2-t_1))\\
D^{(j)}(t_1,t_2)=\calN_j^D(\int_{t_1}^{t_2}C_a(\tau)\mu_j\dd \tau),
\een
where $\calN_j^A(z)$ and $\calN_j^D(z)$ denote independent Poisson random variables with rate $z$, for all $j \in \calJ$. Recall that $C(t)=\left(C_a^{(j)}(t),C_v^{(j)}(t)\right)$ where $C_v^{(j)}$ is the set of virtual templates (i.e., the templates that do not contain jobs of type $j$) and $C_a^{(j)}$ is the set of actual templates.
It is easy to see that
\ben
|f(Q^{(j)}(t+u))-f(Q^{(j)}(t))|\leq f^\prime(0)|Q^{(j)}(t+u)-Q^{(j)}(t)|,
\een
by the mean value theorem, and the fact that $f$ is a concave increasing function.
For notational compactness, let $ \mathbb{E}_{\ve{S}(t)}[\cdot]= \mathbb{E}[\cdot|\ve{S}(t)]$, where $\ve{S}(t)$ is the state of the system at each time $t$. Clearly the the maximum number of templates that can fit in a configuration is less than $M$. It is easy to see that
\begin{align*}
&\expectsT{V(t+u)-V(t)}  \leq \\
&\sum_j \frac{1}{\mu_j} \expectsT{f(Q^{(j)}(t))(A^{(j)}(t,t+u)-D^{(j)}(t,t+u) }\\
&+K_2u+\oo(u),
\end{align*}
%where
%\ben
%K_1&\leq&f^\prime(0)\sum_j\expect{\left(A^{(j)}(t,t+h)\right)^2+\left(D^{(j)}(t,t+h)\right)^2}\\
%&\leq &f^\prime(0)\sumj(\lambda_jh(\lambda_jh+1)+N_{max}\mu_{j}h(N_{max}\mu_{j}h+1))\\
%&=&K_2h+\oo(h)
%\een
where $ K_2=f^\prime(0)\sumj (\rho_j+M)$. Clearly virtual templates do not exist in the configuration if there are more than $M$ jobs in the system. Hence it follows that
%Note that $D^{(j)}(t,t+h)$ counts the number of departures of templates containing jobs (the so-called actual templates), in interval $(t,t+h)$. Hence,
%\ben
%\sumj \expectsT{D^{(j)}(t,t+h)f(Q^{(j)}(t))} &=& \sumj \expectsT{f(Q^{(j)}(t))\mu_jh\left(|C^{(j)}(t)| -|C_v^{(j)}(t)|\right)}-\oo(h)
%\een
%Clearly virtual templates do not exist if there are more than $N_{max}$ jobs in the system, i.e.,
%\ben
%|C_v^{(j)}(t)| f(Q^{(j)}(t))&=&|C_v^{(j)}(t)| f(Q^{(j)}(t)) \mathds{1}\{Q^{(j)}(t)\leq N_{max}\}\\
%& \leq & N_{max} f(N_{max}).
%\een
%Therefore,
\begin{align*}
&\sumj \frac{1}{\mu_j}\expectsT{D^{(j)}(t,t+u)f(Q^{(j)}(t))}\geq \\
&  {\sumj  \expectsT{|C^{(j)}(t)|f(Q^{(j)}(t))}u}-K_3u-\oo(u)
\end{align*}
for $K_3=Mf(M)$. Note that the algorithm keeps the configuration fixed over intervals $[kT,(k+1)T)$, i.e., $C(t)=C(kT)$ for $t \in[kT,(k+1)T)$. Let $\Delta_{t,u}:=\frac{1}{u}\expectsT{V(t+u)-V(t)}$, then,
\ben
{\Delta_{t,u}} &\leq & \expectsT{\sumj f(Q^{(j)}(t)) \left(\rho_j  - |C^{(j)}(kT)|\right)} \\
&&+ K_2+K_3+\oo(1)\\
& \leq& \sumj f(Q^{(j)}(kT)\Big(\rho_j  - |C^{(j)}(kT)|\Big)\\
&&+K_2^2T+K_2+K_3+\oo(1)
\een
%For any $t \in[kT,(k+1)T)$,
%\ben
%|f(Q^{(j)}(t))-f(Q^{(j)}(kT))|&\leq& f^\prime(0)|Q^{(j)}(t)-Q^{(j)}(kT)|\leq K_2T,
%\een
%therefore,
%%\ben
%%\expectsT{|C^{(j)}(t)|\mu_jf(Q^{(j)}(t))}&\geq&|C^{(j)}(kT)|\mu_jf\left(Q^{(j)}(kT)\right) -K_4T,
%%\een
%%where $K_4=f^\prime(0)N_{max}^2\mu_{max}^2$.
%%Let $\Delta_{t,h}:=\frac{V(t+h)-V(t)}{h}$. Then
%\ben
%{\Delta_{t,h}} &\leq& \sumj \rho_jf(Q^{(j)}(kT)  -\sumj |C^{(j)}(kT)|\mu_jf(Q^{(j)}(kT))+K_2^2T+K_2+K_3+\oo(1)
%\een
Taking the limit $u \to 0$,
\ben
\frac{\dd \expectsT{V(t)}}{\dd t} \leq \sumj f(Q^{(j)}(kT)\Big(\rho_j  - |C^{(j)}(kT)|\Big)+K_4
\een
where $K_4=K_2^2T+K_2+K_3$.
Let $\Delta_k = \expectsT{V((k+1)T)-V(kT)}$, then over the $k$-th cycle
\ben
\frac{\alpha}{T} \Delta_{k} +\expectsT{\sumj \sum_{ A \in C^{(j)}(kT)}  b^{(j)}_{A}} \leq \mbox{Term}_1 - \mbox{Term}_2 +\alpha K_4, \label{eq:terms}
\een
where
\ben
\mbox{Term}_1&=&\alpha \sumj \rho_j f(Q^{(j)}(kT)),\\
\mbox{Term}_2&=&\sumj \sum_{ A \in C^{(j)}(kT)} \left(\alpha f(Q^{(j)}(kT)) - b^{(j)}_{A}\right).
\een
%\be
%&&\frac{\alpha}{T} \Delta_{k} +\expectsT{\sumj \sum_{ A \in C^{(j)}(kT)}  b^{(j)}_{A}} \leq \nonumber \\
%%&\alpha \sumj \lambda_j f(Q^{(j)}(t))- \expects{  \sumj \sum_{ i=1}^{ N_j(t)} \alpha \mu_jf(Q^{(j)}(t)) -\sumj \sum_{ i=1}^{ N_j(t)}  b^{(j)}_{A_i}}+\alpha ((K_2+K_3)+\oo(1))=\\
%&& \alpha \sumj \lambda_j f(Q^{(j)}(kT)) -\sumj \sum_{ A \in C^{(j)}(kT)} \left(\alpha \mu_jf(Q^{(j)}(kT)) - b^{(j)}_{A}\right)+\alpha K_4. \label{eq:property}
%\ee
Let $x^\star$ be the optimal solution to the static partitioning problem. The rest of the proof is similar to the Lyapunov analysis in the proof of Theorem~\ref{th3} (step 3), i.e., for any $0 \leq \delta \leq \delta^\star$,
%By the feasibility of $x^\star$, $\rho_j \leq \sum_{A\in \calA^{(j)}}x^\star_A$, for all $j \in \calJ$, hence
\begin{align*}
&\mbox{Term}_1 \leq \alpha \sumj \sum_{A \in A^{(j)}}{{x_A^\star}} f(Q^{(j)}(kT))\\
&\mbox{Term}_2 \leq \sumj \sum_{A \in \calA^{(j)}} \Big((1+\delta){x^\star}_{A}\Big) \Big(\alpha f(Q^{(j)}(kT)-b^{(j)}_{A}) \Big)
\end{align*}
%By assumption, $\rho$ is strictly inside the capacity region and there exists a $\delta^\star$ such that for all $0 \leq \delta \leq \delta^\star$, $\rho(1+\delta) \in \Lambda$. It is easy to show by the monotonicity of $\calC$ (i.e., if $C\in \calC$, $C \backslash A \in \calC$, for all $A \in C^{(j)}$, $j \in \calJ$) that, at the optimal solution, the constraint \dref{cons1} should in fact hold with equality. Hence $\left(\sum_{A \in \calA^{(j)}}{x^\star}_{A}(1+\delta ): j \in \calJ \right) \in \Lambda$ for any $0 \leq \delta \leq \delta^\star$. The Frame-Based algorithm maximizes $\mbox{Term}_2$ over all $C \in \calC$, therefore it follows that
Putting everything together,
%\ben
%&&\frac{\alpha}{T} \Delta_{k} +\expectsT{\sumj \sum_{A \in C^{(j)}}  b^{(j)}_{A}} \leq \\
%&& \alpha \sumj \sum_{A \in A^{(j)}}\mu_j {{x_A^\star}} f(Q^{(j)}(kT))- \sumj \sum_{A \in \calA^{(j)}} \left((1+\delta){x^\star}_{A}\right) \left(-b^{(j)}_{A}+\alpha \mu_{j}f(Q^{(j)}(kT))  \right)  +\alpha K_4 \leq \\
%&&(1+{\delta}) G(x^\star) -\alpha\delta\lambda_{min}\sumj f(Q^{(j)}(kT))+\alpha K_4
%\een
\ben
\frac{\alpha}{T} \Delta_{k} +\expectsT{\sumj \sum_{A \in C^{(j)}}  b^{(j)}_{A}} \leq (1+{\delta}) G(x^\star) \\
-\alpha\delta\rho_{min}\sumj f(Q^{(j)}(kT))+\alpha K_4.
\een
Then it follows from the Foster-Lyapunov theorem that the Markov chain $ve{S}(kT)$, $k=0,1,2,\cdots$ (and therefore Markov chain $\ve{S}(t), t \geq 0$) is positive recurrent. As in the step 3 in the proof of of Theorem~\ref{th3}, we take the expectation from both sides of the above equality with respect to $\ve{S}(kT)$, and then sum over $k=0,..,N-1$, divide by $N$, and let $N\to \infty$.
%yields
%\be
%\limsup_M \frac{1}{M}\sum_{k=0}^{M-1}\sumj\expect{f(Q^{(j)}(kT))} &\leq& \frac{\alpha K_4+(1+\delta)G(x^\star)}{\alpha \delta\lambda_{min}}\\
%\limsup_M \frac{1}{M}\sum_{k=0}^{M-1}\expect{\sumj \sum_{A \in C^{(j)}(kT)}b^{(j)}_A}& \leq &(1+\delta)G(x^\star)+\alpha K_4.
%\ee
%The above inequalities can be independently optimized over $\delta \in [0, \delta^\star]$ (the performance of the algorithm is independent of $\delta$). Then we get
%\be
%\limsup_M \frac{1}{M}\sum_{k=0}^{M}\sumj\expect{Q^{(j)}(kT)} &\leq& \frac{\alpha K_4+(1+\delta^\star)G(x^\star)}{\alpha \delta^\star \lambda_{min}},\\
%\limsup_M \frac{1}{M}\expect{\sumj \sum_{A \in C^{(j)}(kT)}  b^{(j)}_{A}}&\leq &G(x^\star)+\alpha K_4
%\ee
%Consequently,
%\be
%\limsup_t \frac{1}{t}\int_0^{t}\sumj\expect{f(Q_{j}(\tau))}\dd \tau &\leq& \frac{K_4+(1+\delta^\star)G(x^\star)/\alpha}{ \delta^\star \lambda_{min}}+K_2T,\\
%\limsup_t \frac{1}{t}\int_0^t\expect{\sumj \sum_{A \in C^{(j)}(\tau)}  b^{(j)}_{A}}\dd \tau&\leq &G(x^\star)+\alpha K_4.
%\ee
Then the statement of the theorem follows by choosing $B_1=K_2+K_3$ and $B_2=K_2^2$.

%%%%%%%%%%%%%%%%%%%%%%%%%%%%%%%%%%%%%%%%%%%%%%%%%%%%%%%%%%%%%%%%%%%%%%%%%%%%%%%%%%%%%%%%%555
\section{An Alternative Description of Dynamic Graph Partitioning Algorithm}
The Dynamic Graph Partitioning (DGP) algorithm, as described in Section~\ref{sec:low}, does not require any dedicated clock as the decisions are made at the instances of job arrival and departure. In this section, we present an alternative description of the algorithm by using dedicated clocks. Each queue $Q^{(j)}$ is assigned an independent Poisson clock of rate $\hat{\lambda}e^{\alpha f^{(j)}( \ve{\bias}+\ve{Q}(t))/\beta}$, where $\hat{\lambda}$ is a fixed constant depending on how fast the iterations in the algorithm can be performed. Equivalently, at each time $t$, the time duration until the tick of the next clock is an exponential random variable with parameter $\hat{\lambda}e^{\alpha f^{(j)}( \ve{\bias}+\ve{Q}(t))/\beta}$. This means if $Q^{(j)}$ changes at time $t^\prime >t$ before the clock makes a tick, the time duration until the next tick is reset to an independent exponential random variable with parameter $\hat{\lambda}e^{\alpha f^{(j)}( \ve{\bias}+\ve{Q}(t^\prime))/\beta}$. The description of the algorithm is given below.

%\vspace{0.15in}
%\hrule
%\vspace{0.03in}
%Basic Low Complexity Algorithm
%\subsubsection*{Alternative Dynamic Graph Partitioning Algorithm \\$\mathrm{(ADGP)}$}
\begin{algorithm}
\caption{Alternative Dynamic Graph Partitioning $\mathrm{(ADGP)}$ Algorithm}
%\vspace{0.03in}
%\hrule
%\vspace{0.03in}
\noindent \textbf{At the instances of dedicated clocks.\\}
Suppose the dedicated clock of queue $Q^{(j)}$ makes a tick, then:
\begin{algorithmic}[1]
\STATE A virtual template $A^{(j)}$ is chosen randomly from currently feasible templates for graph $\calG_j$, given the current configuration, using \textit{Random Partition Procedure}, if possible. Then this template is added to the configuration with probability $e^{-\frac{1}{\beta}b^{(j)}_A}$ and discarded otherwise. The virtual template leaves the system after an exponentially distributed time duration with mean $1/\mu_j$.
\STATE If there is a job of type $j$ in $Q^{(j)}$ waiting to get service, and a virtual template of type $j$ is created in step 1, this virtual template is filled by a job from $Q^{(j)}$ which converts the virtual template to an actual template.
\end{algorithmic}

\noindent \textbf{At arrival instances.}
\begin{algorithmic}[1]
\STATE Suppose a graph (job) of type $\calG_j$ arrives. The job is added to queue $Q^{(j)}$.
%\STATE If there are virtual templates of type $j$ available in the system, the job is placed in one of the virtual templates chosen arbitrarily at random.
\end{algorithmic}

\noindent \textbf{At departure instances.}

\begin{algorithmic}[1]
\STATE At the departure instances of actual/vitual templates, the algorithm removes the corresponding template from the configuration.
\STATE If this is a departure of an actual template, the job is departed and the corresponding queue is updated.
\end{algorithmic}
\end{algorithm}
%\vspace{0.03in}
%\hrule
%\vspace{0.15in}

%Each queue $Q^{(j)}$ has a dedicated Poisson clock of rate $\hat{\lambda} e^{\alpha f(Q^{(j)}(t))/\beta}$ \footnote{Equivalently, at each time $t$, the time duration until the tick of the clock is an exponential random variable with rate $\hat{\lambda}e^{\alpha f( Q^{(j)}(t))/\beta}$. If $Q^{(j)}$ changes at time $t^\prime >t$ before the clock makes a tick, the time duration until the tick is reset to an independent exponential random variable with rate $\hat{\lambda}e^{\alpha f( Q^{(j)}(t^\prime))/\beta}$.}. At the epoch of the  $j$-th clock, a template $A$ is chosen uniformly at random among currently feasible templates for graph $\calG_j$, given the current configuration (using the procedure described in the Basic Low-Complexity Algorithm 1), then the template is added to the configuration with probability $e^{-\frac{1}{\beta}b^{(j)}_A}$ and discarded otherwise. If there are jobs of type $j$ waiting for service, the virtual templates of type $j$ are filled by these jobs arbitrarily at random. At the departure instances of actual/vitual templates, the algorithm removes the corresponding template from the configuration.

The algorithm will yield average queue size and partitioning cost performance similar to those in Theorem \ref{th3}. The proof essentially follows the three steps of the proof of Theorem~\ref{th3}. Here, we only describe the main property of the Alternative Dynamic Graph Partitioning algorithm with fixed weights, which we refer to as $\overline{\mathrm{ADGP}}(\tilde{W})$ (the counterpart $\overline{\mathrm{DGP}}(\tilde{W})$ in Section \ref{sec:proofs}).
%\subsection{Extension to the robust low complexity algorithm}

\begin{proposition} \label{th4}
Under $\overline{\mathrm{ADGP}}(\tilde{W})$, the steady state distribution of configurations solves
\ben
\max_{\pi}\expectp{\sum_{j\in \calJ} \sum_{A \in C^{(j)}} \tw}-\beta D_{KL}(\pi \parallel \hat{\gamma})
\een
for the following distribution $\hat{\gamma}$
\be \label{gammaform2}
\hat{\gamma}_C=\frac{1}{Z_{\hat{\gamma}}}\left(\sum_\ell |m_l|-\sum_j|C^{(j)}||V_j|\right) !\prod_{j}{\left(\frac{\hat{\lambda}}{\mu_j}\right)}^{|C^{(j)|}}
\ee
where $\hat{Z}_{\hat{\gamma}}$ is the normalizing constant.
\end{proposition}
Similarly to the $\mathrm{DGP}(\tilde{W})$ algorithm, as $\beta \to \infty$, the optimizing $\pi$ converges to $\hat{\gamma}$. The distribution $\hat{\gamma}$ has the interpretation of the steady state distribution of configurations in a loss system with arrival rates $\hat{\lambda_j}=\hat{\lambda}$, $j \in \calJ$, and service rates $\hat{\mu_j}=\mu_j$, $j \in \calJ$. In the loss system, when a graph arrives, it is randomly distributed over the machines if possible; otherwise it is dropped. At the departure instances, the job and hence its template leave the system.
\begin{proof}[of Proposition \ref{th4}]
The proof is basically identical to the proof of Proposition~\ref{th2}. The only difference is that the detailed balance equations are given by
\ben
\pi(C \oplus A^{(j)})\mu_j =\pi(C) \frac{ \hat{\lambda}e^{\alpha f(Q^{(j)})/\beta}}{|\calA^{(j)}(C)|}e^{-\frac{1}{\beta}b^{(j)}_A}
\een
for any configuration $C$ and $C \oplus A^{(j)} \in \calC$; $A^{(j)}\in \calA^{(j)}$,
$j \in \calJ$. Here the LHS is the departure rate of (virtual or actual) template $A^{(j)}$ from the configuration $C \oplus A^{(j)}$. The RHS is the rate at which the (actual or virtual) template $A^{(j)}$ for graphs $\calG_j$ is added to configuration $C$ (the \textit{Random partition Procedure} selects a template $A \in \calA^{(j)}(C)$ uniformly at random).
Thus the detailed balanced equations are given by
\be
\pi(C \oplus A^{(j)})&=&\pi(C)\frac{\hat{\lambda}/\mu_j}{{|\calA^{(j)}(C)|}}e^{\frac{1}{\beta}w^{(j)}_A}.
\ee
%where
%$$|\calA^{(j)}(C)|=\binom{\sum_\ell |m_\ell|-\sum_j |C^{(j)}||V_j|}{|V_j|}|V_j|\, !.$$
%It is then easy to show by induction that
%\ben
%\pi(C)&=&\pi(0)\frac{(\sum_\ell |m_\ell|-\sum_j|C^{(j)}||V_j|)\, !}{(\sum_\ell |m_\ell|)\, !}\prod_{j}\left(\frac{\hat{\lambda}}{\mu_j}\right)^{|C^{(j)}|}\\
%&& \times \exp\left(\frac{1}{\beta}\sum_{j \in \calJ} \sum_{ A \in C^{(j)}} w^{(j)}_{A}\right)\\
%&=& \frac{1}{\hat{Z}_\beta}\hat{\gamma}_C\exp\left(\frac{1}{\beta}\sum_{j \in \calJ} \sum_{ A \in C^{(j)}} w^{(j)}_{A}\right)
%\een
 and it is easy to see that $\dref{expform}$ with $\gamma$ replaced with $\hat{\gamma}$ in (\ref{gammaform2}), indeed satisfies the detailed balance equations, with the normalizing condition that $\sum_C \pi(C)=1$.
%which is the steady state distribution generated by the algorithm, with $\hat{\gamma}$ defined in (\ref{gammaform2}), and $\hat{Z}_\beta=Z_{\hat{\gamma}}\pi(0)$.
The fact that that this distribution maximizes the stated objective function follows in parallel with the arguments in the proof of Proposition \ref{th2}.
%As $\beta \to \infty$, $\hat{Z}_\beta \to \hat{Z}_\infty$ and $D(\pi,\hat{\gamma})\to 0$.
\end{proof}
%The algorithm will yield queue size and partitioning cost performances identical to those in Theorem \ref{th3}, with $\gamma$ being replaced with $\hat{\gamma}$. The details are omitted. 

\end{document}